\documentclass[12pt]{extarticle}
\usepackage{geometry}
\usepackage{subfig}
\makeatletter
\newcommand\ackname{Acknowledgements}
\if@titlepage
\newenvironment{acknowledgements}{
	\titlepage
	\null\vfil
	\@beginparpenalty\@lowpenalty
	\begin{center}
		\bfseries \ackname
		\@endparpenalty\@M1
\end{center}}
{\par\vfil\null\endtitlepage}
\else

\usepackage{mathrsfs}
\usepackage{eucal}
\usepackage{amssymb,amsmath}

\usepackage{pgf}
\usepackage{tikz}
\usepackage{tikz-3dplot}
\usetikzlibrary{decorations.pathreplacing,calligraphy,positioning}
\usepackage{cite}
\usepackage{graphicx}
\usepackage{caption}
\usepackage{color}
\usepackage{amssymb}
\usepgflibrary{arrows}
\usetikzlibrary{calc,angles,positioning,intersections,quotes,backgrounds,decorations,decorations.markings,decorations.text,decorations.pathreplacing}
\usepackage{amsmath}
\usepackage{amsthm}
\usepackage{accents} 
\newtheorem{theorem}{Theorem}[section]
\newtheorem*{remark}{Remark}
\theoremstyle{remark}
\theoremstyle{definition}
\newtheorem{definition}{Definition}[section]
\usepackage{thmtools, thm-restate}

\numberwithin{equation}{section}

\usepackage{comment}
\usepackage{physics}
\definecolor{mediumtaupe}{rgb}{0.4, 0.3, 0.28}
\definecolor{skobeloff}{rgb}{0.0, 0.48, 0.45}
\definecolor{sandstorm}{rgb}{0.93, 0.84, 0.25}
\definecolor{khaki}{rgb}{0.76, 0.69, 0.57}
\definecolor{olivedrab7}{rgb}{0.24, 0.2, 0.12}
\definecolor{sanddune}{rgb}{0.59, 0.44, 0.09}
\definecolor{mediumseagreen}{rgb}{0.24, 0.7, 0.44}
\definecolor{persianplum}{rgb}{0.44, 0.11, 0.11}

\DeclareSymbolFont{rsfs}{U}{rsfs}{m}{n}
\DeclareSymbolFontAlphabet{\mathscrsfs}{rsfs}
\makeatletter
\newcommand*{\rom}[1]{\expandafter\@slowromancap\romannumeral #1@}
\makeatother

\usetikzlibrary{decorations.pathreplacing,decorations.markings}
\tikzset{
	on each segment/.style={
		decorate,
		decoration={
			show path construction,
			moveto code={},
			lineto code={
				\path [#1]
				(\tikzinputsegmentfirst) -- (\tikzinputsegmentlast);
			},
			curveto code={
				\path [#1] (\tikzinputsegmentfirst)
				.. controls
				(\tikzinputsegmentsupporta) and (\tikzinputsegmentsupportb)
				..
				(\tikzinputsegmentlast);
			},
			closepath code={
				\path [#1]
				(\tikzinputsegmentfirst) -- (\tikzinputsegmentlast);
			},
		},
	},
	mid arrow/.style={postaction={decorate,decoration={
				markings,
				mark=at position .5 with {\arrow[#1]{stealth}}
	}}},
}
\begin{document}
	\title{\textbf{Quantum integrability: Lagrangian 1-form case}}
	\author{Thanadon Kongkoom$^1$ and  Sikarin Yoo-Kong$^{2,*}$ \\
		\small {The Institute for Fundamental Study (IF),} \\ \small\emph{Naresuan University, Phitsanulok, Thailand, 65000.}\\
		\small{$^1$thanadonko63@nu.ac.th, $^2$sikariny@nu.ac.th}, $^*$Corresponding author: 0000-0002-6604-8686\\
	}
	\date{}
	\maketitle
	\abstract
	A new notion of integrability called the multi-dimensional consistency for the integrable systems with the Lagrangian 1-form structure is captured in the geometrical language for quantum level. A zero-curvature condition, which implies the multi-dimensional consistency, will be a key relation, e.g. Hamiltonian operators. Therefore, the existence of the zero-curvature condition directly leads to the path-independent feature of the mapping, e.g. multi-time evolution in the Schr\"{o}dinger picture. Another important result is the formulation of the continuous multi-time propagator. With this new type of the propagator, a new perspective on summing all possible paths unavoidably arises as not only all possible paths in the space of dependent variables but also in the space of independent variables must be taken into account. The semi-classical approximation is applied to the multi-time propagator expressing in terms of the classical action and the fluctuation around it. Therefore, the extremum propagator, resulting in path independent feature on the space of independent variables, would guarantee the integrability of the system.
	\section{Introduction}
	Classically, the standard notion of integrability of the Hamiltonian systems is the Liouville-Arnold theorem \cite{Babelon, Arnoldtextbook}. In this notion of integrability, the Hamiltonian systems, whose the evolution is given on $2N$-dimensional manifold called the cotangent bundle, must possess $N$ invariant quantities which are independent and in involution. A key feature in this context is the Hamiltonian commuting flows as a direct consequence of the involution. Alternatively, the integrability can be inferred from the existence of the r-matrix, which is equivalent with the involution relation, through the language of the Lax matrices \cite{Babelon}.
	\\
	\\
	In the discrete context, the standard Liouville-Arnold theorem can be constructed \cite{vesalop}. However, the discrete world is quite fascinate in the sense that all variables are treated on the same equal footing. Consequently, there are various notions of integrability, e.g. existence of r-matrix \cite{Lax}, singularity confinement \cite{Grammaticos} and algebraic entropy \cite{Bellon}. However, there is one remarkable aspect of integrable multi-dimensional discrete systems known as a multi-dimensional consistency \cite{Frankbook}. With this feature, one can consistently express the difference equations in a multi-dimensional lattice, i.e., two dimensional lattice system can be consistently embedded in a three dimensional lattice such that the quadrilateral equations describing three side-to-side connected surfaces of a cube can be solved for a coincide result with a given initial conditions \cite{S.Lobb}. This feature is known as the consistency-around-the cube (CAC). Later, Adler, Bobenko and Suris employed this property to classify the quadrilateral equations for two dimensional lattice known as the ABS list \cite{Adler}. Another important aspect of the multi-dimensional consistency in the discrete level is the Lagrangian multi-form theory. Lobb and Nijhoff first set out to formulate the discrete theory for 2-form and 3-form cases \cite{S.Lobb, S.Lobb2}. A key relation in this context is the Lagrangian closure relation, which holds on the solution of the system, as a direct result of the variation of the action with respect to independent variables. The existence of the Lagrangian closure relation guarantees the constant value of the action under local deformation of the surface in the 2-form case and the volume in the 3-form case on the space of independent variables. Soon later, the 1-from case was formulated by Yoo-Kong, Lobb and Nijhoff \cite{RinCM} in both discrete and continuous levels through an important model known as the Calogero-Moser system \cite{CM, FrankCM}, see also \cite{RinRS}. Again, the existence of the Lagrangian 1-form closure relation guarantees the constant value of the action under local deform of the curve on the space of independent variables. Therefore, the feature implies path-independent property and the multi-time evolution does not depends on the choice of paths, but rather the end points on the space of independent variables. Indeed, this is nothing but the multi-dimensional consistency feature\footnote{One needs to include the variation with respect to dependent variables resulting in the generalised Euler-Lagrange equations and constraints. These equations all together give us a compatible system of equations describing the multi-time evolution of the system.} which is represented in the level of Lagrangians\footnote{As we mentioned earlier that the multi-dimensional consistency was first formulated on the level of discrete equations of motion.}. After these pioneer works, a series of papers has been producing and pushing further in various aspects as well as various systems \cite{Xenitidis, Suris1, Boll1, Boll2, Jairuk1, Jairuk2, Jairuk3, Suris2, Kels1, Petrera, Kels2, Mats1, Mats2, Caudrelier, Mats3}.
	\\
	\\
	In quantum realm, the notion of integrability is not well established. Naively, one can follow the canonical quantisation by promoting a set of invariances or a set of Hamiltonians to be a set of Hamiltonian operators. Therefore, the integrability demands commutator of the Hamiltonian operators to be zero\footnote{This can be viewed as the quantum analogue of the involution.}. However, Weigert \cite{Stefan} provided an encounter example, which is non-integrable system, satisfying the vanishing commutator condition. However, many attempts have been put further to investigate quantum integrability on demanding a \textit{quantum correction} terms $\hbar^2$ \cite{correct1, correct2, correct3}, promoted from the invariances of the counterpart of classical system and commutations of them, see in \cite{QuIS1}. In discrete level, a key tool to study quantum integrable system is the quantum mapping was established in \cite{QuMap} and was applied in the integrability context in \cite{IntQuMap1, IntQuMap2}. Alternatively, Feynman approach on quantising the system might be a better choice \cite{FeynmanHibbs}. The pioneer works on this direction were investigted by Field and Nijhoff \cite{Field}, see also \cite{Fieldthesis} in the discrete systems. Recently, King and Nijhoff set out to formulate quantum path integration incorporated with the quadratic Lagrangian multi-form structure in the discrete level \cite{SD.Kings}. What they did is to impose the periodic reduction on linearised discrete KdV in one particular direction, resulting in the discrete harmonic oscillator. Imposing on another discrete direction, one obtains another discrete harmonic oscillator. Therefore, the discrete Lagrangians for these harmonic oscillators can be explicitly written. Consequently, the explicit form of the multi-discrete propagator can be obtained, since the Gaussian integral can be used in this case. An intriguing feature of this multi-discrete propagator is path-independent feature on the space of independent discrete variables. In other words, the multi-discrete propagator remains the same under local deformation of the discrete paths on the space of independent discrete variables. Then, in this quantum scenario, one might have to take all possible discrete paths not only on the space of dependent variables, but also on the space of independent variables into account. This new conceptual view on propagator was first proposed by Nijhoff \cite{Franktalk} in the continuous case to capture the multi-dimensional consistency in language of the path integrals. However, the explicit connection between the discrete set up and the Nijhoff's continuous proposal for the Lagrangian multi-form of the propagator is still missing. In this contribution, the multi-time propagator for the case of arbitrary Lagrangian 1-forms is constructed and integrability condition is given through the scheme called the quantum variational principle. 
	\\
	\\
	The structure of this paper is as follows. In section \ref{subsection3.1}, a set of multi-time Schr\"{o}dinger equations is given and the consistency of these equations will be derived resulting in the zero-curvature condition for the Hamiltonian operators. Then the unitary multi-time operator is expressed in terms of the Wilson line and this unitary multi-time operator possesses the path-independent as a direct result of the zero-curvature condition of the Hamiltonian operators. In section \ref{subsection3.2}, the multi-time propagator is systematically derived. A new type of the functional measure is defined, since all possible path on the space dependent variables and independent variables must be included. The semi-classical method is applied and then integrability, indicating by path-independent feature on the space of independent variables, of the system shall be given based on the closure relation of Lagrangians.
 In section \ref{section5}, the summary together with important remarks will be given.
	\section{Multi-dimensional consistency as integrability}\label{section3}
	Let us first recall the two main features used to indicate the integrability of the classical system. For the Hamiltonian systems with  $N$ degrees of freedom, there exists a set of first integrals (treated as Hamiltonains) $\{H_1,H_2,...,H_N\}$, which are independent and in involution: $\{H_l,H_k\}=0\;,\;l\neq k=1,2,...,N$, or equivalently, $\partial H_l/\partial t_k = \partial H_k/\partial t_l\;,\;l\neq k=1,2,...,N$. Consequently, the involution leads to an important feature known as the Hamiltonian commuting flows. On the Lagrangian side, there also exists a set of Lagrangian 1-form $\{L_1,L_2,...,L_N\}$, which satisfies the relations: $\partial L_l/\partial t_k = \partial L_k/\partial t_l\;,\;l\neq k=1,2,...,N$ known as the closure relation. Both involution and closure relation imply the multi-dimensional consistency of the multi-time evolution of the system. At this point, it is very natural to extend the idea to the quantum level. In the standard manner, with the existence of the Hamiltonian, one can lift to the quantum case with the Schr\"{o}dinger approach. With the existence of the Hamiltonian hierarchy, one expects to have a set of Schr\"{o}dinger equations and therefore, the compatible multi-time evolution of the wave function is what we are interested. Alternatively, one can prefer to quantise the system with Feynman approach based on the action and Lagrangian. In the same analogy, there exists the Lagrangian hierarchy. Then, one also expects to have multi-time propagators with a compatible multi-time evolution.
	\subsection{Schr\"{o}dinger picture and zero-curvature condition}\label{subsection3.1}
	In this section, the multi-time evolution of the wave function will be studied. To derive a set of Schr\"{o}dinger equations, we shall promote a set of non-autonomous Hamiltonians $\{H_1,H_2,...,H_N \}$, associated with time variables $\mathbf{t}\equiv\{t_1,t_2,...,t_N \}$, to be a set of non-autonomous Hamiltonian operators $\{\hat{\mathbf{H}}_1,\hat{\mathbf{H}}_2,...,\hat{\mathbf{H}}_N \}$, resulting in
	\begin{equation}
	    i\hbar\frac{\partial }{\partial t_l} \ket{\Psi(\mathbf{t})}= \hat{\mathbf{H}}_l\ket{\Psi(\mathbf{t})}\;,\;l=1,2,...,N\;,\label{Sch1}
	\end{equation}
	where $\left|\Psi(\mathbf{t})\right\rangle$ is the eigenstate for $\hat{\mathbf{H}}_l,\;l=1,2,...,N$. 
	\begin{theorem}
	Let $\ket{\Psi(\mathbf{t})}$ be a multi-time normalised vector in a Hilbert space $\mathcal{H}$ and  $\{\hat{\mathbf{H}}_l:\mathcal{H}\mapsto\mathcal{H}\;,\;l=1,2,\dots,N\}$
	be a set of non-autonomous Hamiltonian operators. All Schr\"{o}dinger equations in \eqref{Sch1} will be consistent if 
	\begin{equation}
	    \frac{\partial \hat{\mathbf{H}}_k}{\partial t_l}-\frac{\partial \hat{\mathbf{H}}_l}{\partial t_k}-\frac{i}{\hbar}\big[\hat{\mathbf{H}}_k,\hat{\mathbf{H}}_l\big]=0\;,\;k\neq l=1,2,...,N\;,
	\end{equation}
	hold.
	\end{theorem}
	\begin{proof}
	The quantities $\partial_{t_l}-(1/i\hbar)\hat{\mathbf{H}}_l$, where $l=1,2,...,N$, can be treated as the covariant derivative and the system of equations \eqref{Sch1} is overdetermined. Thus, $\hat{\mathbf{H}}_l$, where $l=1,2,...,N$, must satisfy a compatible condition such that
	\begin{equation}
	     \frac{\partial^2}{\partial t_l\partial t_k}\ket{\Psi(\mathbf{t})} = \frac{\partial^2}{\partial t_k\partial t_l}\ket{\Psi(\mathbf{t})}\;.\label{compatH}
	\end{equation}
	This \eqref{compatH} is nothing but equivalently
	\begin{equation}
	    \frac{\partial \hat{\mathbf{H}}_k}{\partial t_l}-\frac{\partial \hat{\mathbf{H}}_l}{\partial t_k}-\frac{i}{\hbar}\big[\hat{\mathbf{H}}_k,\hat{\mathbf{H}}_l\big]=0\;,\;k \neq l =1,2,\dots,N\;,\label{Sch0curve}
	\end{equation}
	which is therefore the zero-curvature condition for the Hamiltonian operators.
	\\
	\\
	Alternatively, if we define a unitary operator for the multi-time evolution $\hat {\mathbf{U}}(\mathbf{t})$ such that
	\begin{equation}
	    \ket{\Psi(\mathbf{t})} = \hat{\mathbf{U}}(\mathbf{t})\ket{\Psi(0)}\;,
	\end{equation}
	where $\hat{\mathbf{U}}^\dagger(\mathbf{t})\hat{\mathbf{U}}(\mathbf{t})=\mathbb{I}$ with $\langle \Psi(\mathbf{t})|\Psi(\mathbf{t})\rangle=\langle \Psi(0)|\Psi(0)\rangle$, the equations \eqref{Sch1} give a set of equations
	\begin{equation}
	    i\hbar\frac{\partial}{\partial t_l}\hat{\mathbf{U}}(\mathbf{t}) = \hat{\mathbf{H}}_l \hat{\mathbf{U}}(\mathbf{t})\;,\;l=1,2,...,N\;.\label{unitarytime1}
	\end{equation}
	If we are interested in the unitary evolution associated with time variable $t_l$, we obtain
	\begin{equation}
	    \hat{\mathbf{U}}_l(\mathbf{t})=\hat{\mathbf{U}}_l(0,0,\dots,t_l,\dots,0,0)=\mathrm{T}e^{-\frac{i}{\hbar}\int_0^{t_l}\hat{\mathbf{H}}_ldt_l}\;,
	\end{equation}
	where $\mathrm{T}$ is time-ordering operator.
		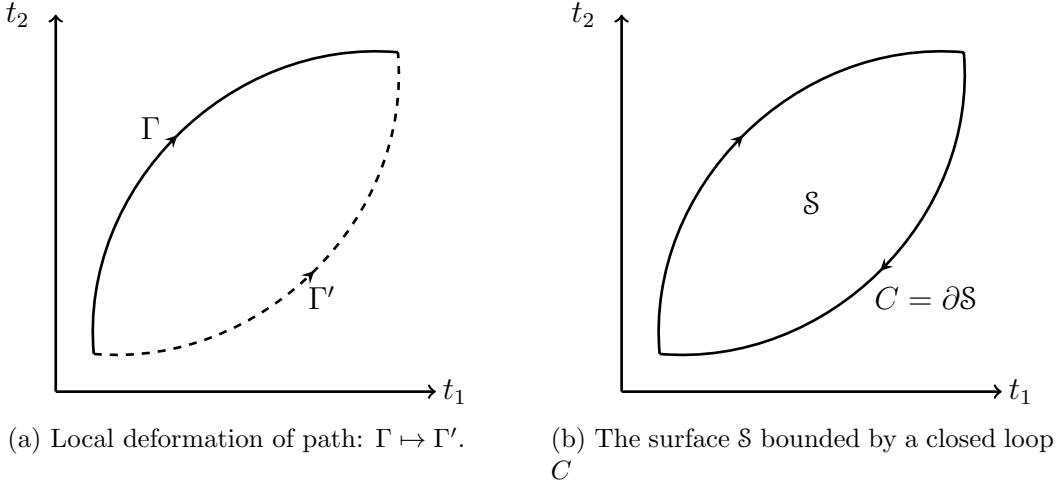
\begin{figure}[h]
    	    \centering
    	    \subfloat[Local deformation of path: $\Gamma\mapsto\Gamma'$.]{\label{deformation of path}
        	    \begin{tikzpicture}[]
            	     \path[line width = 1pt, draw = black,->]
 	                  (0,0) to (5,0);
 	                  
 	                  \path[line width = 1pt, draw = black,->]
 	                  (0,0) to (0,5);
 	                  
 	                 \path[line width = 1pt, draw = black, postaction ={on each segment = {mid arrow}}]
 	                  (0.5,0.5) to [bend left = 50] (4.5,4.5);
 	                 
 	                 \path[dashed, line width = 1pt, draw = black, postaction ={on each segment = {mid arrow}}]
 	                  (0.5,0.5) to [bend right = 50] (4.5,4.5);
 	                  
 	                 \node at (5.25,0) {$t_1$};
 	                 \node at (-0.5,5) {$t_2$};
 	                 
 	                 \node at (3.5,1.25) {$\Gamma'$};
 	                 \node at (1.25,3.5) {$\Gamma$};
 	                 
 	            \end{tikzpicture}
 	       }
 	       \qquad
 	        \subfloat[The surface $\mathcal{S}$ bounded by a closed loop $C$]{\label{loop path}
        	    \begin{tikzpicture}[]
        	        \path[line width = 1pt, draw = black,->]
 	                  (0,0) to (5,0);
 	                  
 	                  \path[line width = 1pt, draw = black,->]
 	                  (0,0) to (0,5);
 	                  
 	                 \path[line width = 1pt, draw = black, postaction ={on each segment = {mid arrow}}]
 	                  (0.5,0.5) to [bend left = 50] (4.5,4.5)
 	                  (4.5,4.5) to [bend left = 50] (0.5,0.5);
 	                  
 	                 \node at (5.25,0) {$t_1$};
 	                 \node at (-0.5,5) {$t_2$};
 	                 
 	                 \node at (4,1.25) {$C=\partial \mathcal{S}$};
 	                 \node at (2.5,2.5) {$\mathcal{S}$};
 	                
 	            \end{tikzpicture}
 	         }
 	         \caption{The evolution of system on the space of 2 time variables.}
 	         \label{deform and loop}
    \end{figure}
	Therefore, on the space of time variables, the unitary operator, which maps the state along the path $\Gamma$ shown in figure \ref{deformation of path}, is given by\footnote{The construction of the composite map is given in appendix \ref{AppendixA}.}
	\begin{equation}
	    \hat{\mathbf{U}}_\Gamma(\mathbf{t}) =\hat{\mathbf{U}}_1\circ\hat{\mathbf{U}}_2\dots\circ\hat{\mathbf{U}}_N = \mathrm{T}e^{-\frac{i}{\hbar}\int_\Gamma\sum_{j=1}^N\hat{\mathbf{H}}_jdt_j}\;.\label{U12}
	\end{equation}
	Equation \eqref{U12} is nothing but the Wilson line representation of the operator $\hat{\mathbf{U}}$ and the Hamiltonian operators will be treated as the gauge variables in this situation. Under the local deformation: $\Gamma\mapsto \Gamma'$, the evolution is consistent : $\hat{\mathbf{U}}_\Gamma = \hat{\mathbf{U}}_{\Gamma'}$. Therefore, this relation implies the path-independent feature of the multi-time evolution of the state $|\Psi\rangle$. Equivalently, this path-independent feature gives identity for the loop evolution, see figure \ref{loop path}, such that
	\begin{equation}
        \hat{\mathbf{U}}_\circlearrowleft(\mathbf{t}) = \mathrm{T}e^{-\frac{i}{\hbar}\oint_{C=\partial \mathcal{S}}\sum_{j=1}^N \hat{\mathbf{H}}_j dt_j }=\mathrm{P}e^{-\frac{i}{\hbar}\iint_{\mathcal{S}}\sum_{k\geqslant l}^N\sum_{l=1}^N\hat{\mathbf{Z}}_{lk} dt_l\wedge dt_k }=\mathbb{I}\;,
    \end{equation}
    where $\mathrm{T}$ is a time-ordering operator for a loop path\footnote{The definition of the time-ordering operator is given in appendix \ref{Appendix0}.} and $\mathrm{P}$ is a surface-ordering operator\cite{Arefeva, aa} and
	\begin{equation}
	    \hat{\mathbf{Z}}_{lk} = \hat{\mathbf{U}}_\gamma^{-1}\bigg(\frac{\partial \hat{\mathbf{H}}_k}{\partial t_l}-\frac{\partial \hat{\mathbf{H}}_l}{\partial t_k}-\frac{i}{\hbar}\big[\hat{\mathbf{H}}_k,\hat{\mathbf{H}}_l\big]\bigg)\hat{\mathbf{U}}_\gamma=0
	\end{equation}
	is a twisted curvature \cite{aa} in the Schr\"{o}dinger picture where $\gamma$ represents an arbitrary path connecting between the origin and $\mathbf{t}$. The vanishing curvature is nothing but the equation \eqref{Sch0curve}. We finally note that this compatible multi-time evolution is actually a commutation of multi-time unitary operators: $[\hat{\mathbf{U}}_l,\hat{\mathbf{U}}_k]=0$. Therefore, this commutativity can be treated as a quantum integrability in the Schr\"{o}dinger's picture. 
	\end{proof}
	\noindent
    We note that the relation \eqref{Sch0curve} is a quantum version of the involution of the Hamiltonians in the classical context
    \[
    \frac{\partial {{H}}_k}{\partial t_l}-\frac{\partial {{H}}_l}{\partial t_k}+\big\{{{H}}_k,{{H}}_l\big\}=0\;,
    \]
    which gives a core feature of the classical integrability known as the Hamiltonian commuting flows.
    
	\begin{remark}
	The integrable quantum systems must possess the zero-curvature condition, but the converse is not necessary true\cite{Stefan}.
	\end{remark}
	\subsection{Feynman picture and compatible multi-time propagators}\label{subsection3.2}
	Let us first briefly outline what we are going to do in this section. We shall first give derivation with a general setup and later fit our results with a recent development on quantum multi-dimensional consistency \cite{SD.Kings}. Therefore, we would like to provide some key points of King and Nijhoff work. What they did is investigate the two dimensional periodic reduction of the lattice KdV, resulting in two different discrete-time harmonic oscillators. By doing redefinition of discrete parameters, one obtains two different continuous harmonic oscillators
	\begin{eqnarray}
	     \frac{d^2\mathbf{q}}{dt_1^2}+\omega_1^2\mathbf{q}=0\;,\label{quad1}\\
	     \frac{d^2\mathbf{q}}{dt_2^2}+\omega_2^2\mathbf{q}=0\;,\label{quad2}
	\end{eqnarray}
	where $\mathbf{q}(t_1,t_2)$ is position variables and $(t_1,t_2)$ are two different time variables. Here the mass of the system is set to be one. The variables $\omega_1$ and $\omega_2$ play the role of frequency for the first and second systems, respectively. Then, they study properties of the multi-time propagator(in the discrete setup) for the quadratic Lagrangian 1-form case. The King-Nijhoff formula for multi-discrete propagator 
	\begin{equation}
	    K(\mathbf{q}_b(\mathcal{M},\mathcal{N});\mathbf{q}_a(0,0))=\sum_{\Gamma\in \mathscr{B}}\mathscr{N}_\Gamma K_\Gamma(\mathbf{q}_b(\mathcal{M},\mathcal{N});\mathbf{q}_a(0,0) 
	\end{equation}
	possesses the path independent feature on the space of independent variables, see figure \ref{discrete path}. Then loops will not contribute to the propagator as a direct result of Lagrangian closure relation. Here $\mathscr {N}_\Gamma$ is a normalising factor and $\mathscr{B}$ is a set of all possible paths in the discrete space of independent variables.
	 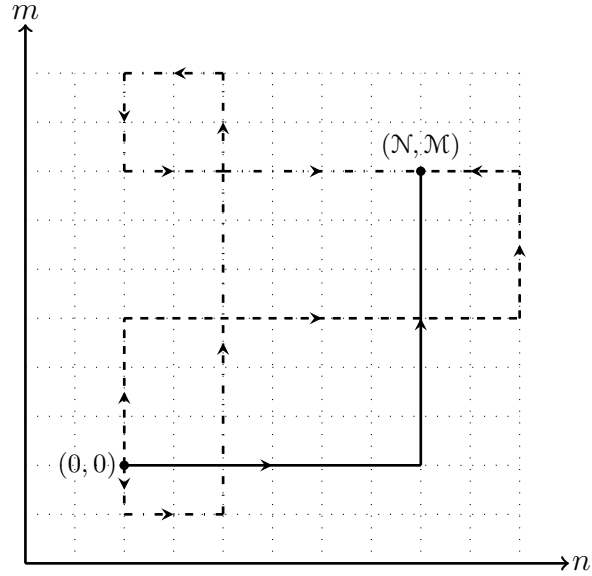
\begin{figure}[h]
    	    \centering
        	    \begin{tikzpicture}[scale = 0.65]
            	    \path[loosely dotted, draw=black]
 	                (0,0) to (10,0)
 	                (10,0) to (10,10)
 	                (10,10) to (0,10)
 	                (0,10) to (0,0)
 	                (2,0) to (2,10)
 	                (4,0) to (4,10)
 	                (6,0) to (6,10)
 	                (8,0) to (8,10)
 	                (1,0) to (1,10)
 	                (3,0) to (3,10)
 	                (5,0) to (5,10)
 	                (7,0) to (7,10)
 	                (9,0) to (9,10)
 	                (0,2) to (10,2)
 	                (0,4) to (10,4)
 	                (0,6) to (10,6)
 	                (0,8) to (10,8)
 	                (0,1) to (10,1)
 	                (0,3) to (10,3)
 	                (0,5) to (10,5)
 	                (0,7) to (10,7)
 	                (0,9) to (10,9);
 	                
 	                \path[->, line width = 1pt, draw=black]
 	                (0,0) to (11,0);
 	                
 	                \path[->, line width = 1pt, draw=black]
 	                (0,0) to (0,11);
 	                
 	                \path[line width = 1pt, draw=black, postaction ={on each segment = {mid arrow}}]
 	                (2,2) to (8,2)
 	                (8,2) to (8,8);
 	                
 	                \path[dashed, line width = 1pt, draw=black, postaction ={on each segment = {mid arrow}}]
 	                (2,2) to (2,5)
 	                (2,5) to (10,5)
 	                (10,5) to (10,8)
 	                (10,8) to (8,8);
 	                
 	                \path[loosely dashdotted, line width = 1pt, draw=black, postaction ={on each segment = {mid arrow}}]
 	                (2,2) to (2,1)
 	                (2,1) to (4,1)
 	                (4,1) to (4,8)
 	                (4,8) to (4,10)
 	                (4,10) to (2,10)
 	                (2,10) to (2,8)
 	                (2,8) to (4,8)
 	                (4,8) to (8,8);
 	                
 	                \node at (11.25,0) {$n$};
 	                \node at (0,11.25) {$m$};
 	                
 	                \node at (1.25,2) {\footnotesize{$(0,0)$}};
 	                \node at (8,8.5) {\footnotesize{$(\mathcal{N},\mathcal{M})$}};
 	              
                    \node[circle, fill, inner sep=1.25 pt] at (2,2) {};
                    \node[circle, fill, inner sep=1.25 pt] at (8,8) {};
 	                
 	            \end{tikzpicture}
 	        \captionof{figure}{Some possible paths on the space of independent discrete variables.}\label{discrete path}
    \end{figure}
	What we will do next is to extend their result to the continuous case. We first would like to provide some basic ingredients for the standard propagator(the single-time propagator) given by
	\begin{equation}
	    K\left(\mathbf{q}'',t'';\mathbf{q}',t'\right)=:\int_{\mathbf{q}'}^{\mathbf{q}''}\mathscr{D}[\mathbf{q}(t)]e^{\frac{i}{\hbar}S[\mathbf q(t)]} \;.\label{standard propagator}
	\end{equation}
	Here $S[\mathbf q(t)]=\int_{t'}^{t''}L(\mathbf{q},\dot{\mathbf{q}};t)dt$ is the action functional and $L(\mathbf{q},\dot{\mathbf{q}};t)$ is the standard Lagrangian. The notation $\int_{\mathbf{q}'}^{\mathbf{q}''}\mathscr{D}[\mathbf{q}(t)]\;\label{path functional}$ plays the role of the functional measure over the configuration space of the paths (spatial paths). The role of the propagator is to map the state from $\braket{\mathbf{q}'}{\Psi(t')}$ to $\braket{\mathbf{q}''}{\Psi(t'')}$.
	\\
	\\
    In practice, the propagator \eqref{standard propagator} is not convenient for the calculation. Therefore, the semi-classical approximation method is applied by considering all constructive contributions around the classical path. To proceed such method, we write $\mathbf{q}(t)=\mathbf{q}_c(t) + \mathbf{y}(t)$, see figure \ref{q+y}, where $\mathbf{q}_c(t)$ is the classical path and $\mathbf{y}(t)$ is a fluctuation with the boundary conditions: $\mathbf{y}(t'')=\mathbf{y}(t') = 0$.
	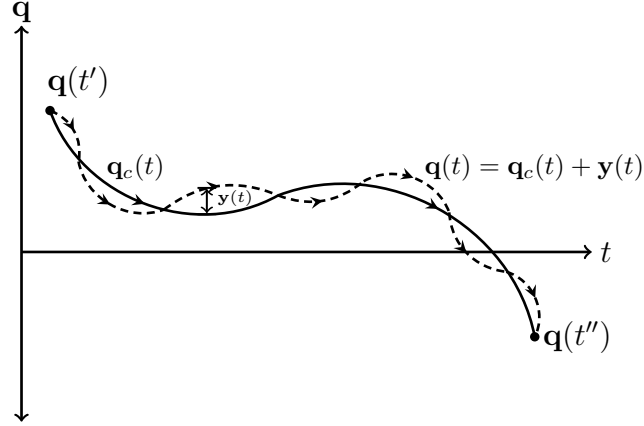
\begin{figure}[h]
        \centering
        	 \begin{tikzpicture}[scale = 0.75]
            	    \path[->, line width = 1 pt, draw = black]
	                (0,0) to (10,0);
	                
	                \path[->, line width = 1 pt, draw = black]
	                (0,0) to (0,-3);
	                
	                \path[->, line width = 1 pt, draw = black]
	                (0,0) to (0,4);
	                
	                \path[line width = 1 pt, draw = black, postaction ={on each segment = {mid arrow}}]
	                (0.5,2.5) to [bend right = 50] (4.5,1)
	                (4.5,1) to [bend left = 50] (9,-1.5);
	                
	                \path[densely dashed, line width = 1 pt, draw = black, postaction ={on each segment = {mid arrow}}]
	                (0.5,2.5) to [bend left = 40] (1,1.6)
	                (1,1.6) to [bend right = 50] (2.5,0.75)
	                (2.5,0.75) to [bend left = 30] (4.5,1)
	                (4.5,1) to [bend right = 30] (6,1.25)
	                (6,1.25) to [bend left = 50] (7.5,0.7)
	                (7.5,0.7) to [bend right = 40] (8.5,-0.35)
	                (8.5,-0.35) to [bend left = 50] (9,-1.5);
	                
	                \draw[|<->|,line width = 0.75 pt] (3.25,1.15) --  (3.25,0.65);
	                
	                \node[circle, fill, inner sep=1.25 pt] at (0.5,2.5) {};
	                \node[circle, fill, inner sep=1.25 pt] at (9,-1.5) {};
	                
	                \node at (3.75,0.95) {\tiny{$\mathbf{y}(t)$}};
	                \node at (2,1.5) {\footnotesize{$\mathbf{q}_c(t)$}};
	                \node at (9,1.5) {\footnotesize{$\mathbf{q}(t)=\mathbf{q}_c(t)+\mathbf{y}(t)$}};
	                \node at (0,4.25) {$\mathbf{q}$};
	                \node at (10.25,0) {$t$};
	                \node at (9.75,-1.5) {$\mathbf{q}(t'')$};
	                \node at (1,3) {$\mathbf{q}(t')$};
	                
 	          \end{tikzpicture}
        \caption{The fluctuation around the classical path $\mathbf{q}_c(t)$ with the initial point $\mathbf{q}(t')$ and the final point $\mathbf{q}(t'')$.}
        \label{q+y}
    \end{figure}
    Doing the expansion with respect to the fluctuation, the propagator becomes
	\begin{eqnarray}
	     K\left(\mathbf{q}'',t'';\mathbf{q}',t'\right) 
         = e^{\frac{i}{\hbar}S[\mathbf{q}_c(t)]}\mathcal{Q}(\mathbf{q}'',t'',\mathbf{q}',t')\left[1+\mathcal{O}(\hbar)\right]\;,\label{semiclpath}	     
	\end{eqnarray}
    where
     \begin{eqnarray}
	     \mathcal{Q}(\mathbf{q}'',t'',\mathbf{q}',t') = 
	     \int_{\mathbf{q}'}^{\mathbf{q}''}\mathscr{D}[\mathbf{y}(t)]e^{\frac{i}{2\hbar}\int_{t'}^{t''}d\tau\int_{t'}^{t''}d\sigma\left(\mathbf{y}(\tau)\frac{\delta^2S[\mathbf{q}_c(t)]}{\delta\mathbf{q}(\tau)\delta\mathbf{q}(\sigma)}\mathbf{y}(\sigma)\right)}\;.\label{factor for single path}
	\end{eqnarray}
	Here $\mathcal{Q}(\mathbf{q}'',t'',\mathbf{q}',t')$ is a smooth function of end points since the variable $\mathbf{y}$ is integrated out. The explicit form the function $\mathcal{Q}(\mathbf{q}'',t'',\mathbf{q}',t')$ can be formally expressed as \cite{condensebook, feynmanhandbook}
	\begin{equation}
	    \mathcal{Q}(\mathbf{q}'',t'',\mathbf{q}',t') = \det\left(\frac{i}{2\pi\hbar}\frac{\partial^2S[\mathbf{q}_c(t)]}{\partial\mathbf{q}(t'')\partial\mathbf{q}(t')}\right)^{\frac{1}{2}}\;.
	\end{equation}
	\\
	With all basic ingredients above, we are now ready to extend the notion of the propagator \eqref{standard propagator} into the case of multi-time Lagrangian 1-forms.
	\begin{definition}[A multi-time propagator]
	Let $\mathscr{L}=\sum_{j=1}^NL_jdt_j$ be a Lagrangian 1-form , where $L_j=L_j\left(\mathbf{q},\left\{\frac{\partial\mathbf{q}}{\partial t_j};j = 1, 2,\dots,N\right\};\mathbf{t}\right)$. On the space of independent variables (time variables) parameterised by a variable $s$ such that $\mathbf{t}(s)$, where $s'<s<s''$, the multi-time propagator is given by
	\begin{equation}
	     K(\mathbf{q}(\mathbf{t}(s'')),s''; \mathbf{q}(\mathbf{t}(s')), s') =\int_{\mathbf{q}(\mathbf{t}(s'))}^{\mathbf{q}(\mathbf{t}(s''))}\mathbb{D}[\mathbf{q}(\mathbf{t}(s));\Gamma\in\mathscr{B}]e^{\frac{i}{\hbar}\int_{\{\Gamma:\Gamma\in\mathscr{B}\}}\mathscr{L}}\;,\label{Multipro}
	\end{equation}
	where $\mathscr{L}=ds\sum_{j=1}^NL_jdt_j/ds$ and  $\int\mathbb{D}[\mathbf{q}(s);\Gamma\in\mathscr{B}]$ is the functional measure over all possible spatial-temporal paths. Here $\Gamma\in\mathscr{B}$, where $\mathscr{B}$ is a family of paths connecting the point $\mathbf{t}(s')$ with the point $\mathbf{t}(s'')$ on space of time variables.
	\end{definition}
	\begin{figure}[h]
    	    \centering
        	    \begin{tikzpicture}[scale = 0.55]
            	    \path[dashed, draw=black]
 	                (0,0) to (10,0)
 	                (10,0) to (10,10)
 	                (10,10) to (0,10)
 	                (0,10) to (0,0)
 	                (2,0) to (2,10)
 	                (4,0) to (4,10)
 	                (6,0) to (6,10)
 	                (8,0) to (8,10)
 	                (1,0) to (1,10)
 	                (3,0) to (3,10)
 	                (5,0) to (5,10)
 	                (7,0) to (7,10)
 	                (9,0) to (9,10)
 	                (0,2) to (10,2)
 	                (0,4) to (10,4)
 	                (0,6) to (10,6)
 	                (0,8) to (10,8)
 	                (0,1) to (10,1)
 	                (0,3) to (10,3)
 	                (0,5) to (10,5)
 	                (0,7) to (10,7)
 	                (0,9) to (10,9);
 	                
 	                \path[line width = 1pt, draw=black]
 	                (-0.5,-0.5) to (10.5,10.5);
 	                
 	                \node at (10,11.25) {\scriptsize{symmetrical}};
 	                \node at (10,10.75) {\scriptsize{dividing line}};
 	                
 	                \draw [-stealth](0,-0.5) -- (2,-0.5);
 	                \draw [-stealth](-0.5,0) -- (-0.5,2);
 	                
 	                \node at (-0.75,1) {\scriptsize{$t_2$}};
 	                \node at (1,-0.75) {\scriptsize{$t_1$}};
 	                
 	                \draw [decorate, decoration = {brace,mirror}] (10.25,0) --  (10.25,10);
 	                \draw [decorate, decoration = {calligraphic brace}] (0,10.25) --  (10,10.25);
 	                
 	                \draw [decorate, decoration = {brace,mirror}] (6.25,0) --  (6.25,1);
 	                \draw [decorate, decoration = {calligraphic brace,mirror}] (5,-0.25) --  (6,-0.25);
    
 	                \node at (10.75,5) {\scriptsize{$\mathcal{N}$}};
 	                \node at (5,10.75) {\scriptsize{$\mathcal{N}$}};
 	                \node at (6.75,0.5) {\scriptsize{$\epsilon_2$}};
 	                \node at (5.5,-0.75) {\scriptsize{$\epsilon_1$}};
 	            \end{tikzpicture}
 	        \captionof{figure}{The two dimensional space of time variables $t_1$ and $t_2$ partitioning by the plaquettes with the size $\epsilon_1\times\epsilon_2$}.\label{lattice of 2-time structure}
    \end{figure}
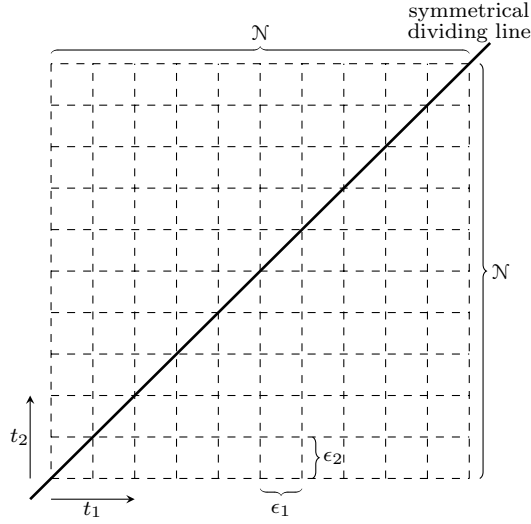
    \noindent To construct the multi-time propagator \eqref{Multipro}, we shall first start with the case of 2 time variables $(t_1,t_2)$ for simplicity. We therefore partition space into $\mathcal{N}\times \mathcal{N}$\footnote{One can consider the space without symmetry. However, at the end, we are going to consider the limit $\mathcal{N}\rightarrow\infty$. Therefore, it is simpler in terms of formulation with the symmetric case.}, see figure \ref{lattice of 2-time structure}. We also keep in mind that the continuum limit is already taken into account. This means that the propagator in each line element, e.g. $(i,j)-(i+\epsilon_1,j)$, where $i,j\in [0,N]$, can be automatically expressed in the form: $K=\int_{\mathbf{q}(i,j)}^{\mathbf{q}(i+\epsilon_1,j)}\mathscr{D}[\mathbf{q}(t_1,j)]e^{\frac{i}{\hbar}\int_{(i,j)}^{(i+\epsilon_1,j)}L_1(t_1,j)dt_1}$. Moreover, we will consider only the forward time steps (as we did in the standard single time case) and we will also employ the symmetry of the lattice by considering the first evolution in $t_1$ together with all possible deformations and later $t_2$ together with all possible deformations, see figure \ref{6 formats}. To illustrate the idea, let's consider first the simplest path $\Gamma\in\mathscr{B}_1$, where $\mathscr{B}_1$ is a set of the path configurations\footnote{In this case, there is only one path configuration.}, shown in figure \ref{fig:formatA}. The propagator is given by
    \begin{align}
        K^{(1)} = \int_{\mathbf{q}(0,0)}^{\mathbf{q}(\epsilon_1\mathcal{N},0)}\mathscr{D}[\mathbf{q}(t_1,0)]\int_{-\infty}^{\infty} d^Nq(\epsilon_1\mathcal{N},0)\int_{\mathbf{q}(\epsilon_1\mathcal{N},0)}^{\mathbf{q}(\epsilon_1\mathcal{N},\epsilon_2\mathcal{N})}\mathscr{D}[\mathbf{q}(\epsilon_1\mathcal{N},t_2)]e^{\frac{i}{\hbar}\Big(\int_{(0,0)}^{(\epsilon_1\mathcal{N},0)}L_1(t_1,0)dt_1+ \int_{(\epsilon_1\mathcal{N},0)}^{(\epsilon_1\mathcal{N},\epsilon_2\mathcal{N})}L_2(\epsilon_1\mathcal{N},t_2)dt_2\Big)},\label{K^1}
    \end{align}
    where the superscript $(1)$ denotes the first simplest path. The propagator in \eqref{K^1} contains all spatial paths along $(0,0)-(\mathcal{N},0)$ and all spatial paths along $(\mathcal{N},0)-(\mathcal{N},\mathcal{N})$. Both sections are glued by the completeness term $\int_{-\infty}^{\infty} d^Nq$ at the corner $(\mathcal{N},0)$. We note that, in equation \eqref{K^1}, the normalising factor is dropped out for our convenient.
    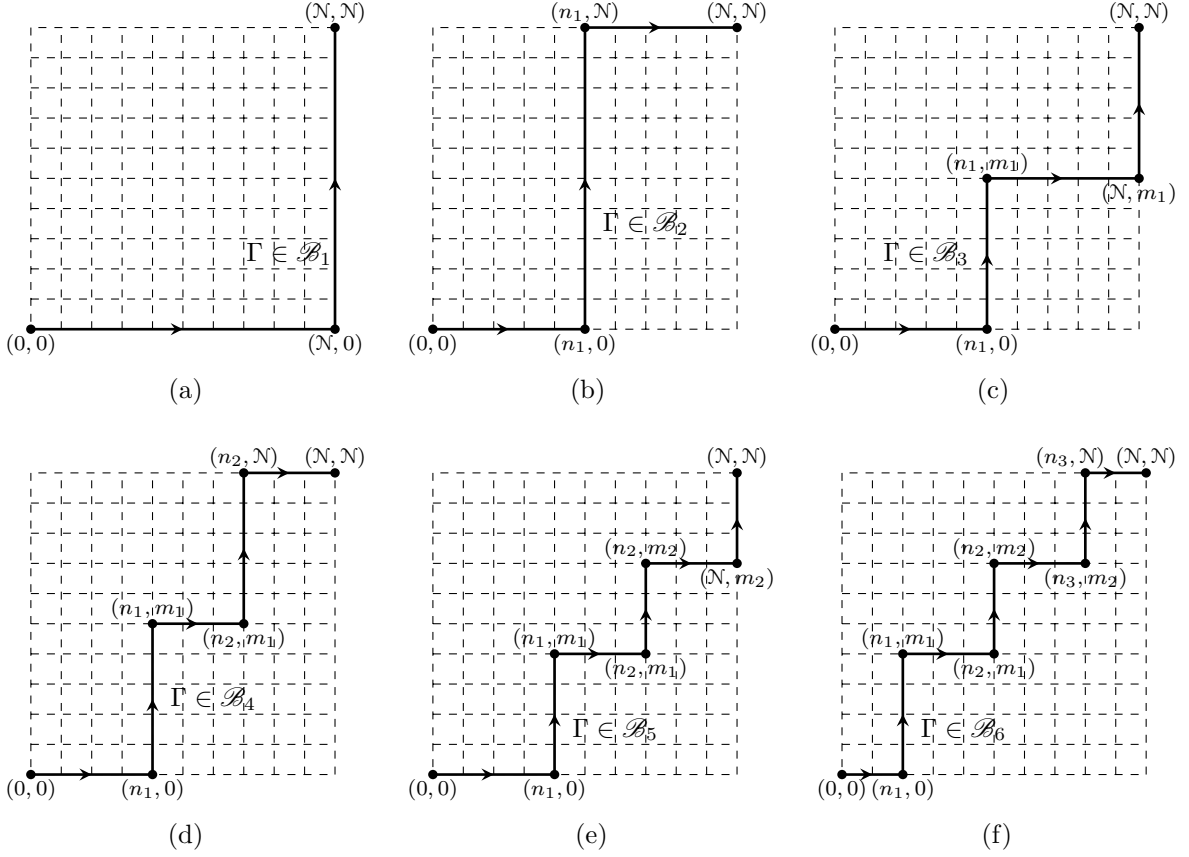
\begin{figure}[h]
    	    \centering
    	    \subfloat[]{\label{fig:formatA}
        	    \begin{tikzpicture}[scale = 0.4]
            	    \path[dashed, draw=black]
 	                (0,0) to (10,0)
 	                (10,0) to (10,10)
 	                (10,10) to (0,10)
 	                (0,10) to (0,0)
 	                (2,0) to (2,10)
 	                (4,0) to (4,10)
 	                (6,0) to (6,10)
 	                (8,0) to (8,10)
 	                (1,0) to (1,10)
 	                (3,0) to (3,10)
 	                (5,0) to (5,10)
 	                (7,0) to (7,10)
 	                (9,0) to (9,10)
 	                (0,2) to (10,2)
 	                (0,4) to (10,4)
 	                (0,6) to (10,6)
 	                (0,8) to (10,8)
 	                (0,1) to (10,1)
 	                (0,3) to (10,3)
 	                (0,5) to (10,5)
 	                (0,7) to (10,7)
 	                (0,9) to (10,9);
 	                
 	                \node[circle, fill, inner sep=1.25 pt] at (0,0) {};
 	                \node[circle, fill, inner sep=1.25 pt] at (10,10) {};
 	                
 	                \node at (0,-0.5) {\scriptsize{$(0,0)$}};
 			        \node at (10,10.5) {\scriptsize{$(\mathcal{N},\mathcal{N})$}};
 	                
 	                \path[line width = 1 pt, draw=black, postaction ={on each segment = {mid arrow}}]
 	                (0,0) to (10,0)
 	                (10,0) to (10,10);
 	                
 	                \node[circle, fill, inner sep=1.25 pt] at (10,0) {};
 	                \node at (10,-0.5) {\scriptsize{$(\mathcal{N},0)$}};
 	                \node at (8.5,2.5) {\footnotesize{$\Gamma\in\mathscr{B}_1$}};
 	            \end{tikzpicture}
 	       }
 	        \subfloat[]{\label{fig:formatB}
        	    \begin{tikzpicture}[scale = 0.4]
            	    \path[dashed, draw=black]
 	                (0,0) to (10,0)
 	                (10,0) to (10,10)
 	                (10,10) to (0,10)
 	                (0,10) to (0,0)
 	                (2,0) to (2,10)
 	                (4,0) to (4,10)
 	                (6,0) to (6,10)
 	                (8,0) to (8,10)
 	                (1,0) to (1,10)
 	                (3,0) to (3,10)
 	                (5,0) to (5,10)
 	                (7,0) to (7,10)
 	                (9,0) to (9,10)
 	                (0,2) to (10,2)
 	                (0,4) to (10,4)
 	                (0,6) to (10,6)
 	                (0,8) to (10,8)
 	                (0,1) to (10,1)
 	                (0,3) to (10,3)
 	                (0,5) to (10,5)
 	                (0,7) to (10,7)
 	                (0,9) to (10,9);
 	                
 	                \node[circle, fill, inner sep=1.25 pt] at (0,0) {};
 	                \node[circle, fill, inner sep=1.25 pt] at (10,10) {};
 	                
 	                \node at (0,-0.5) {\scriptsize{$(0,0)$}};
 			        \node at (10,10.5) {\scriptsize{$(\mathcal{N},\mathcal{N})$}};
 	                
 	                \path[line width = 1 pt, draw=black, postaction ={on each segment = {mid arrow}}]
 	                (0,0) to (5,0)
 	                (5,0) to (5,10)
 	                (5,10) to (10,10);
 	                
 	                \node[circle, fill, inner sep=1.25 pt] at (5,0) {};
 	                \node at (5,-0.5) {\scriptsize{$(n_1,0)$}};
 	                
 	                \node[circle, fill, inner sep=1.25 pt] at (5,10) {};
 	                \node at (5,10.5) {\scriptsize{$(n_1,\mathcal{N})$}};
 	                \node at (7,3.5) {\footnotesize{$\Gamma\in\mathscr{B}_2$}};
 	            \end{tikzpicture}
 	         }
 	         \subfloat[]{\label{fig:formatC}
        	    \begin{tikzpicture}[scale = 0.4]
            	    \path[dashed, draw=black]
 	                (0,0) to (10,0)
 	                (10,0) to (10,10)
 	                (10,10) to (0,10)
 	                (0,10) to (0,0)
 	                (2,0) to (2,10)
 	                (4,0) to (4,10)
 	                (6,0) to (6,10)
 	                (8,0) to (8,10)
 	                (1,0) to (1,10)
 	                (3,0) to (3,10)
 	                (5,0) to (5,10)
 	                (7,0) to (7,10)
 	                (9,0) to (9,10)
 	                (0,2) to (10,2)
 	                (0,4) to (10,4)
 	                (0,6) to (10,6)
 	                (0,8) to (10,8)
 	                (0,1) to (10,1)
 	                (0,3) to (10,3)
 	                (0,5) to (10,5)
 	                (0,7) to (10,7)
 	                (0,9) to (10,9);
 	                
 	                \node[circle, fill, inner sep=1.25 pt] at (0,0) {};
 	                \node[circle, fill, inner sep=1.25 pt] at (10,10) {};
 	                
 	                \node at (0,-0.5) {\scriptsize{$(0,0)$}};
 			        \node at (10,10.5) {\scriptsize{$(\mathcal{N},\mathcal{N})$}};
 	                
 	                \path[line width = 1 pt, draw=black, postaction ={on each segment = {mid arrow}}]
 	                (0,0) to (5,0)
 	                (5,0) to (5,5)
 	                (5,5) to (10,5)
 	                (10,5) to (10,10);
 	                
 	                \node[circle, fill, inner sep=1.25 pt] at (5,0) {};
 	                \node at (5,-0.5) {\scriptsize{$(n_1,0)$}};
 	                
 	                \node[circle, fill, inner sep=1.25 pt] at (5,5) {};
 	                \node at (5,5.5) {\scriptsize{$(n_1,m_1)$}};
 	                
 	                \node[circle, fill, inner sep=1.25 pt] at (10,5) {};
 	                \node at (10,4.5) {\scriptsize{$(\mathcal{N},m_1)$}};
 	                \node at (3,2.5) {\footnotesize{$\Gamma\in\mathscr{B}_3$}};
 	            \end{tikzpicture}
 	         }
 	         \\
 	         \subfloat[]{\label{fig:formatD}
        	    \begin{tikzpicture}[scale = 0.4]
            	    \path[dashed, draw=black]
 	                (0,0) to (10,0)
 	                (10,0) to (10,10)
 	                (10,10) to (0,10)
 	                (0,10) to (0,0)
 	                (2,0) to (2,10)
 	                (4,0) to (4,10)
 	                (6,0) to (6,10)
 	                (8,0) to (8,10)
 	                (1,0) to (1,10)
 	                (3,0) to (3,10)
 	                (5,0) to (5,10)
 	                (7,0) to (7,10)
 	                (9,0) to (9,10)
 	                (0,2) to (10,2)
 	                (0,4) to (10,4)
 	                (0,6) to (10,6)
 	                (0,8) to (10,8)
 	                (0,1) to (10,1)
 	                (0,3) to (10,3)
 	                (0,5) to (10,5)
 	                (0,7) to (10,7)
 	                (0,9) to (10,9);
 	                
 	                \node[circle, fill, inner sep=1.25 pt] at (0,0) {};
 	                \node[circle, fill, inner sep=1.25 pt] at (10,10) {};
 	                
 	                \node at (0,-0.5) {\scriptsize{$(0,0)$}};
 			        \node at (10,10.5) {\scriptsize{$(\mathcal{N},\mathcal{N})$}};
 	                
 	                \path[line width = 1 pt, draw=black, postaction ={on each segment = {mid arrow}}]
 	                (0,0) to (4,0)
 	                (4,0) to (4,5)
 	                (4,5) to (7,5)
 	                (7,5) to (7,10)
 	                (7,10) to (10,10);
 	                
 	                \node[circle, fill, inner sep=1.25 pt] at (4,0) {};
 	                \node at (4,-0.5) {\scriptsize{$(n_1,0)$}};
 	                
 	                \node[circle, fill, inner sep=1.25 pt] at (4,5) {};
 	                \node at (4,5.5) {\scriptsize{$(n_1,m_1)$}};
 	                
 	                \node[circle, fill, inner sep=1.25 pt] at (7,5) {};
 	                \node at (7,4.5) {\scriptsize{$(n_2,m_1)$}};
 	                
 	                \node[circle, fill, inner sep=1.25 pt] at (7,10) {};
 	                \node at (7,10.5) {\scriptsize{$(n_2,\mathcal{N})$}};
 	                \node at (6,2.5) {\footnotesize{$\Gamma\in\mathscr{B}_4$}};
 	            \end{tikzpicture}
 	         }
 	         \subfloat[]{\label{fig:formatE}
        	    \begin{tikzpicture}[scale = 0.4]
            	    \path[dashed, draw=black]
 	                (0,0) to (10,0)
 	                (10,0) to (10,10)
 	                (10,10) to (0,10)
 	                (0,10) to (0,0)
 	                (2,0) to (2,10)
 	                (4,0) to (4,10)
 	                (6,0) to (6,10)
 	                (8,0) to (8,10)
 	                (1,0) to (1,10)
 	                (3,0) to (3,10)
 	                (5,0) to (5,10)
 	                (7,0) to (7,10)
 	                (9,0) to (9,10)
 	                (0,2) to (10,2)
 	                (0,4) to (10,4)
 	                (0,6) to (10,6)
 	                (0,8) to (10,8)
 	                (0,1) to (10,1)
 	                (0,3) to (10,3)
 	                (0,5) to (10,5)
 	                (0,7) to (10,7)
 	                (0,9) to (10,9);
 	                
 	                \node[circle, fill, inner sep=1.25 pt] at (0,0) {};
 	                \node[circle, fill, inner sep=1.25 pt] at (10,10) {};
 	                
 	                \node at (0,-0.5) {\scriptsize{$(0,0)$}};
 			        \node at (10,10.5) {\scriptsize{$(\mathcal{N},\mathcal{N})$}};
 	                
 	                \path[line width = 1 pt, draw=black, postaction ={on each segment = {mid arrow}}]
 	                (0,0) to (4,0)
 	                (4,0) to (4,4)
 	                (4,4) to (7,4)
 	                (7,4) to (7,7)
 	                (7,7) to (10,7)
 	                (10,7) to (10,10);
 	                
 	                \node[circle, fill, inner sep=1.25 pt] at (4,0) {};
 	                \node at (4,-0.5) {\scriptsize{$(n_1,0)$}};
 	                
 	                \node[circle, fill, inner sep=1.25 pt] at (4,4) {};
 	                \node at (4,4.5) {\scriptsize{$(n_1,m_1)$}};
 	                
 	                \node[circle, fill, inner sep=1.25 pt] at (7,4) {};
 	                \node at (7,3.5) {\scriptsize{$(n_2,m_1)$}};
 	                
 	                \node[circle, fill, inner sep=1.25 pt] at (7,7) {};
 	                \node at (7,7.5) {\scriptsize{$(n_2,m_2)$}};
 	                
 	                \node[circle, fill, inner sep=1.25 pt] at (10,7) {};
 	                \node at (10,6.5) {\scriptsize{$(\mathcal{N},m_2)$}};
 	                \node at (6,1.5) {\footnotesize{$\Gamma\in\mathscr{B}_5$}};
 	            \end{tikzpicture}
 	         }
 	         \subfloat[]{\label{fig:formatF}
        	    \begin{tikzpicture}[scale = 0.4]
            	    \path[dashed, draw=black]
 	                (0,0) to (10,0)
 	                (10,0) to (10,10)
 	                (10,10) to (0,10)
 	                (0,10) to (0,0)
 	                (2,0) to (2,10)
 	                (4,0) to (4,10)
 	                (6,0) to (6,10)
 	                (8,0) to (8,10)
 	                (1,0) to (1,10)
 	                (3,0) to (3,10)
 	                (5,0) to (5,10)
 	                (7,0) to (7,10)
 	                (9,0) to (9,10)
 	                (0,2) to (10,2)
 	                (0,4) to (10,4)
 	                (0,6) to (10,6)
 	                (0,8) to (10,8)
 	                (0,1) to (10,1)
 	                (0,3) to (10,3)
 	                (0,5) to (10,5)
 	                (0,7) to (10,7)
 	                (0,9) to (10,9);
 	                
 	                \node[circle, fill, inner sep=1.25 pt] at (0,0) {};
 	                \node[circle, fill, inner sep=1.25 pt] at (10,10) {};
 	                
 	                \node at (0,-0.5) {\scriptsize{$(0,0)$}};
 			        \node at (10,10.5) {\scriptsize{$(\mathcal{N},\mathcal{N})$}};
 	                
 	                \path[line width = 1 pt, draw=black, postaction ={on each segment = {mid arrow}}]
 	                (0,0) to (2,0)
 	                (2,0) to (2,4)
 	                (2,4) to (5,4)
 	                (5,4) to (5,7)
 	                (5,7) to (8,7)
 	                (8,7) to (8,10)
 	                (8,10) to (10,10);
 	                
 	                \node[circle, fill, inner sep=1.25 pt] at (2,0) {};
 	                \node at (2,-0.5) {\scriptsize{$(n_1,0)$}};
 	                
 	                \node[circle, fill, inner sep=1.25 pt] at (2,4) {};
 	                \node at (2,4.5) {\scriptsize{$(n_1,m_1)$}};
 	                
 	                \node[circle, fill, inner sep=1.25 pt] at (5,4) {};
 	                \node at (5,3.5) {\scriptsize{$(n_2,m_1)$}};
 	                
 	                \node[circle, fill, inner sep=1.25 pt] at (5,7) {};
 	                \node at (5,7.5) {\scriptsize{$(n_2,m_2)$}};
 	                
 	                \node[circle, fill, inner sep=1.25 pt] at (8,7) {};
 	                \node at (8,6.5) {\scriptsize{$(n_3,m_2)$}};
 	                
 	                \node[circle, fill, inner sep=1.25 pt] at (8,10) {};
 	                \node at (7.5,10.5) {\scriptsize{$(n_3,\mathcal{N})$}};
 	                \node at (4,1.5) {\footnotesize{$\Gamma\in\mathscr{B}_6$}};
 	            \end{tikzpicture}
 	            }
 	   \caption{The 6 first configurations of possible paths.}\label{6 formats}
    \end{figure}
    \\
    \\
    Next, for the paths given in \ref{fig:formatA} and \ref{fig:formatB}, the propagator can be written as
    \begin{eqnarray}
        K^{(2)} &=& \int_{\mathbf{q}(0,0)}^{\mathbf{q}(\mathcal{N},0)}\mathscr{D}[\mathbf{q}(t_1,0)]\int_{-\infty}^{\infty} d^Nq(\mathcal{N},0)\int_{\mathbf{q}(\mathcal{N},0)}^{\mathbf{q}(\mathcal{N},\mathcal{N})}\mathscr{D}[\mathbf{q}(\mathcal{N},t_2)]e^{\frac{i}{\hbar}\Big(\int_{(0,0)}^{(\mathcal{N},0)}L_1(t_1,0)dt_1+\int_{(\mathcal{N},0)}^{(\mathcal{N},\mathcal{N})}L_2(\mathcal{N},t_2)dt_2\Big)} \nonumber
        \\&&+\sum_{n_1=1}^{\mathcal{N}-1}\int_{\mathbf{q}(0,0)}^{\mathbf{q}(n_1,0)}\mathscr{D}[\mathbf{q}(t_1,0)]\int_{-\infty}^{\infty} d^Nq(n_1,0)\int_{\mathbf{q}(n_1,0)}^{\mathbf{q}(n_1,\mathcal{N})}\mathscr{D}[\mathbf{q}(n_1,t_2)]\int_{-\infty}^{\infty} d^Nq(n_1,\mathcal{N}) \nonumber
        \\&&\times \int_{\mathbf{q}(n_1,\mathcal{N})}^{\mathbf{q}(\mathcal{N},\mathcal{N})}\mathscr{D}[\mathbf{q}(t_1,\mathcal{N})]e^{\frac{i}{\hbar}\Big(\int_{(0,0)}^{(n_1,0)}L_1(t_1,0)dt_1+\int_{(n_1,0)}^{(n_1,\mathcal{N})}L_2(n_1,t_2)dt_2+\int_{(n_1,\mathcal{N})}^{(\mathcal{N},\mathcal{N})}L_1(t_1,\mathcal{N})dt_1\Big)}, \label{2first propagator}
    \end{eqnarray}
    where $(2)$ denotes the first two possible simple paths. Here we note again that, for simplicity, the width parameters $\epsilon_1$ and $\epsilon_2$ of the plaquette are dropped out. What we see is that the exponential terms in \eqref{2first propagator} are not the same for the first term and the second term. However, if we introduce a time parameterised variable $s'<s<s''$ such that $t_1(s)$ and $t_2(s)$, therefore, we have $\mathscr{L}=(L_1dt_1/ds+L_2dt_2/ds)ds$. 
    Here, if we replace the $n_1$ of the second term in \eqref{2first propagator} by $\mathcal{N}$, one would obtain
    \begin{align}
        &\int_{\mathbf{q}(0,0)}^{\mathbf{q}(\mathcal{N},0)}\mathscr{D}[\mathbf{q}(t_1,0)]\int_{-\infty}^{\infty} d^Nq(\mathcal{N},0)\int_{\mathbf{q}(\mathcal{N},0)}^{\mathbf{q}(\mathcal{N},\mathcal{N})}\mathscr{D}[\mathbf{q}(\mathcal{N},t_2)]\int_{-\infty}^{\infty} d^Nq(\mathcal{N},\mathcal{N})\int_{\mathbf{q}(\mathcal{N},\mathcal{N})}^{\mathbf{q}(\mathcal{N},\mathcal{N})}\mathscr{D}[\mathcal{q}(t_1,\mathcal{N})]e^{\frac{i}{\hbar}\int_{\{\Gamma:\Gamma\in\mathscr{B}_1\cup\mathscr{B}_2\}}\mathscr{L}}\nonumber
        \\&= \bra{\mathbf{q}(\mathcal{N},\mathcal{N})}\hat{\mathbf{U}}(\mathcal{N},\mathcal{N};\mathcal{N},\mathcal{N})\hat{\mathbf{U}}(\mathcal{N},\mathcal{N};\mathcal{N},0)\hat{\mathbf{U}}(\mathcal{N},0;0,0)\ket{\mathbf{q}(0,0)}\nonumber
        \\&= \bra{\mathbf{q}(\mathcal{N},\mathcal{N})}\hat{\mathbf{U}}(\mathcal{N},\mathcal{N};\mathcal{N},0)\hat{\mathbf{U}}(\mathcal{N},0;0,0)\ket{\mathbf{q}(0,0)}\nonumber
        \\&= \int_{\mathbf{q}(0,0)}^{\mathbf{q}(\mathcal{N},0)}\mathscr{D}[\mathbf{q}(t_1,0)]\int_{-\infty}^{\infty} d^Nq(\mathcal{N},0)\int_{\mathbf{q}(\mathcal{N},0)}^{\mathbf{q}(\mathcal{N},\mathcal{N})}\mathscr{D}[\mathbf{q}(\mathcal{N},t_2)]e^{\frac{i}{\hbar}\int_{\{\Gamma:\Gamma\in\mathscr{B}_1\cup\mathscr{B}_2\}}\mathscr{L}}\;,\label{identity of propagator1}
    \end{align}
    which is the propagator in \eqref{K^1}. Then the propagator $K^{(2)}$ can be simply reduced to
    \begin{align}
        K^{(2)} =& \sum_{n_1=1}^{\mathcal{N}}\int_{\mathbf{q}(0,0)}^{\mathbf{q}(n_1,0)}\mathscr{D}[\mathbf{q}(t_1,0)]\int_{-\infty}^{\infty} d^Nq(n_1,0)\int_{\mathbf{q}(n_1,0)}^{\mathbf{q}(n_1,\mathcal{N})}\mathscr{D}[\mathbf{q}(n_1,t_2)]\int_{-\infty}^{\infty} d^Nq(n_1,\mathcal{N}) \nonumber
        \\&\times \int_{\mathbf{q}(n_1,\mathcal{N})}^{\mathbf{q}(\mathcal{N},\mathcal{N})}\mathscr{D}[\mathbf{q}(t_1,\mathcal{N})]e^{\frac{i}{\hbar}\int_{\{\Gamma:\Gamma\in\mathscr{B}_1\cup\mathscr{B}_2\}}\mathscr{L}}\;.\label{K2}
    \end{align}
    Here $\mathscr{B}_2$ is a set of the path configurations in figure \ref{6 formats}b. Diagrammatically, what we do in \eqref{K2} is just shifting a vertical line from $(1,0)-(1,\mathcal{N})$ to $(\mathcal{N},0)-(\mathcal{N},\mathcal{N})$. 
    \\
    \\
    Next, we include the third possible path, see figure \ref{fig:formatC}, into the calculation. Now the propagator in this case becomes
    \begin{align}
        K^{(3)} =& \int_{\mathbf{q}(0,0)}^{\mathbf{q}(\mathcal{N},0)}\mathscr{D}[\mathbf{q}(t_1,0)]\int_{-\infty}^{\infty} d^Nq(\mathcal{N},0)\int_{\mathbf{q}(\mathcal{N},0)}^{\mathbf{q}(\mathcal{N},\mathcal{N})}\mathscr{D}[\mathbf{q}(\mathcal{N},t_2)] e^{\frac{i}{\hbar}\Big(\int_{(0,0)}^{(\mathcal{N},0)}L_1dt_1+\int_{(\mathcal{N},0)}^{(\mathcal{N},\mathcal{N})}L_2dt_2\Big)} \nonumber
        \\&+\sum_{n_1=1}^{\mathcal{N}-1}\int_{\mathbf{q}(0,0)}^{\mathbf{q}(n_1,0)}\mathscr{D}[\mathbf{q}(t_1,0)]\int_{-\infty}^{\infty} d^Nq(n_1,0)\int_{\mathbf{q}(n_1,0)}^{\mathbf{q}(n_1,N)}\mathscr{D}[\mathbf{q}(n_1,t_2)]\int_{-\infty}^{\infty} d^Nq(n_1,\mathcal{N})\nonumber
        \\&\times\int_{\mathbf{q}(n_1,\mathcal{N})}^{\mathbf{q}(\mathcal{N},\mathcal{N})}\mathscr{D}[\mathbf{q}(t_1,\mathcal{N})]e^{\frac{i}{\hbar}\Big(\int_{(0,0)}^{(n_1,0)}L_1dt_1+\int_{(n_1,0)}^{(n_1,N)}L_2dt_2+\int_{(n_1,\mathcal{N})}^{(\mathcal{N},\mathcal{N})}L_1dt_1\Big)} \nonumber
        \\&+\sum_{m_1=1}^{\mathcal{N}-1}\sum_{n_1=1}^{\mathcal{N} -1}\int_{\mathbf{q}(0,0)}^{\mathbf{q}(n_1,0)}\mathscr{D}[\mathbf{q}(t_1,0)]\int_{-\infty}^{\infty} d^Nq(n_1,0)\int_{\mathbf{q}(n_1,0)}^{\mathbf{q}(n_1,m_1)}\mathscr{D}[\mathbf{q}(n_1,t_2)]\nonumber
        \\&\times\int_{-\infty}^{\infty} d^Nq(n_1,m_1)\int_{\mathbf{q}(n_1,m_1)}^{\mathbf{q}(\mathcal{N},m_1)}\mathscr{D}[\mathbf{q}(t_1,m_1)]\int_{-\infty}^{\infty} d^Nq(\mathcal{N},m_1)\nonumber
        \\&\times\int_{\mathbf{q}(\mathcal{N},m_1)}^{\mathbf{q}(\mathcal{N},\mathcal{N})}\mathscr{D}[\mathbf{q}(\mathcal{N},t_2)] e^{\frac{i}{\hbar}\Big(\int_{(0,0)}^{(n_1,0)}L_1dt_1+\int_{(n_1,0)}^{(n_1,m_1)}L_2dt_2 +\int_{(n_1,m_1)}^{(\mathcal{N},m_1)}L_1dt_1+\int_{(\mathcal{N},m_1)}^{(\mathcal{N},\mathcal{N})}L_2dt_2\Big)}\;. \label{3first propagator}
    \end{align}
    Here the labeled variables activating in Lagrangians on the exponent terms are omitted. We notice that the third term will be identical with the second term if $m_1=\mathcal{N}$ for every single $n_1$ and will be the first term if we let $m_1=0$. Here comes to a crucial point. The order of summation is matter since one might need to avoid the repeatity of arbitrary $n_1$ at $m_1=0$, see appendix \ref{AppendixB}. Then, the sum over possible $m_1$ comes first and sum over $n_1$ comes second. The propagator $K^{(3)}$ becomes
    \begin{align}
         K^{(3)} =& \sum_{n_1=1}^{\mathcal{N}-1}\sum_{m_1=0}^{\mathcal{N}}\int_{\mathbf{q}(0,0)}^{\mathbf{q}(n_1,0)}\mathscr{D}[\mathbf{q}(t_1,0)]\int_{-\infty}^{\infty} d^Nq(n_1,0)\int_{\mathbf{q}(n_1,0)}^{\mathbf{q}(n_1,m_1)}\mathscr{D}[\mathbf{q}(n_1,t_2)]\int_{-\infty}^{\infty} d^Nq(n_1,m_1)\nonumber
         \\&\times\int_{\mathbf{q}(n_1,m_1)}^{\mathbf{q}(\mathcal{N},m_1)}\mathscr{D}[\mathbf{q}(t_1,m_1)]\int_{-\infty}^{\infty} d^Nq(\mathcal{N},m_1)\int_{\mathbf{q}(\mathcal{N},m_1)}^{\mathbf{q}(\mathcal{N},\mathcal{N})}\mathscr{D}[\mathbf{q}(\mathcal{N},t_2)] e^{\frac{i}{\hbar}\int_{\{\Gamma:\Gamma\in\mathscr{B}_1\cup\mathscr{B}_2  \cup\mathscr{B}_3\}}\mathscr{L}}\;.\label{3first propagator complete}
    \end{align}
    Here $\mathscr{B}_3$ is a set of the path configurations in figure \ref{6 formats}c. Diagrammatically, what we do in \eqref{3first propagator complete} is shifting the horizontal line from $(n_1,0)-(\mathcal{N},0)$ to $(n_1,\mathcal{N})-(\mathcal{N},\mathcal{N})$ for $n_1=[1,\mathcal{N}]$. 
    \\
    \\
    With the structure what we proceed so far, it is now not difficult to see that the propagator $K^{(5)}$, included figures \ref{fig:formatA}-\ref{fig:formatE}, can be expressed in the form
    \begin{align}
        K^{(5)} =& \Bigg(\int_{\mathbf{q}(0,0)}^{\mathbf{q}(\mathcal{N},0)}\mathscr{D}[\mathbf{q}(t_1,0)]\int_{-\infty}^{\infty} d^Nq(\mathcal{N},0)\int_{\mathbf{q}(\mathcal{N},0)}^{\mathbf{q}(\mathcal{N},\mathcal{N})}\mathscr{D}[\mathbf{q}(\mathcal{N},t_2)]\nonumber
        \\&+\sum_{n_1=1}^{\mathcal{N}-1}\int_{\mathbf{q}(0,0)}^{\mathbf{q}(n_1,0)}\mathscr{D}[\mathbf{q}(t_1,0)]\int_{-\infty}^{\infty} d^Nq(n_1,0)\int_{\mathbf{q}(n_1,0)}^{\mathbf{q}(n_1,\mathcal{N})}\mathscr{D}[\mathbf{q}(n_1,t_2)]\int_{-\infty}^{\infty} d^Nq(n_1,\mathcal{N})\int_{\mathbf{q}(n_1,\mathcal{N})}^{\mathbf{q}(\mathcal{N},\mathcal{N})}\mathscr{D}[\mathbf{q}(t_1,\mathcal{N})] \nonumber
        \\&+\sum_{m_1=1}^{\mathcal{N}-1}\sum_{n_1=1}^{\mathcal{N}-1}\int_{\mathbf{q}(0,0)}^{\mathbf{q}(n_1,0)}\mathscr{D}[\mathbf{q}(t_1,0)]\int_{-\infty}^{\infty} d^Nq(n_1,0)\int_{\mathbf{q}(n_1,0)}^{\mathbf{q}(n_1,m_1)}\mathscr{D}[\mathbf{q}(n_1,t_2)]\int_{-\infty}^{\infty} d^Nq(n_1,m_1)\nonumber
        \\&\times\int_{\mathbf{q}(n_1,m_1)}^{\mathbf{q}(\mathcal{N},m_1)}\mathscr{D}[\mathbf{q}(t_1,m_1)]\int_{-\infty}^{\infty} d^Nq(\mathcal{N},m_1)\int_{\mathbf{q}(\mathcal{N},m_1)}^{\mathbf{q}(\mathcal{N},\mathcal{N})}\mathscr{D}[\mathbf{q}(\mathcal{N},t_2)] \nonumber
        \\&+\sum_{n_2=n_1+1}^{\mathcal{N}-1}\sum_{m_1=1}^{\mathcal{N}-1}\sum_{n_1=1}^{\mathcal{N}-2}\int_{\mathbf{q}(0,0)}^{\mathbf{q}(n_1,0)}\mathscr{D}[\mathbf{q}(t_1,0)]\int_{-\infty}^{\infty} d^Nq(n_1,0)\int_{\mathbf{q}(n_1,0)}^{\mathbf{q}(n_1,m_1)}\mathscr{D}[\mathbf{q}(n_1,t_2)]\int_{-\infty}^{\infty} d^Nq(n_1,m_1)\nonumber
        \\&\times\int_{\mathbf{q}(n_1,m_1)}^{\mathbf{q}(n_2,m_1)}\mathscr{D}[\mathbf{q}(t_1,m_2)]\int_{-\infty}^{\infty} d^Nq(n_2,m_1)\int_{\mathbf{q}(n_2,m_1)}^{\mathbf{q}(n_2,\mathcal{N})}\mathscr{D}[\mathbf{q}(n_2,t_2)]\int_{-\infty}^{\infty} d^Nq(n_2,\mathcal{N})\int_{\mathbf{q}(n_2,\mathcal{N})}^{\mathbf{q}(\mathcal{N},\mathcal{N})}\mathscr{D}[\mathbf{q}(t_1,\mathcal{N})] \nonumber
        \\&+\sum_{m_2=m_1+1}^{\mathcal{N}-1}\sum_{n_2=n_1+1}^{\mathcal{N}-1}\sum_{m_1=1}^{\mathcal{N}-2}\sum_{n_1=1}^{\mathcal{N}-2}\int_{\mathbf{q}(0,0)}^{\mathbf{q}(n_1,0)}\mathscr{D}[\mathbf{q}(t_1,0)]\int_{-\infty}^{\infty} d^Nq(n_1,0)\int_{\mathbf{q}(n_1,0)}^{\mathbf{q}(n_1,m_1)}\mathscr{D}[\mathbf{q}(n_1,t_2)]\int_{-\infty}^{\infty} d^Nq(n_1,m_1)\nonumber
        \\&\times\int_{\mathbf{q}(n_1,m_1)}^{\mathbf{q}(n_2,m_1)}\mathscr{D}[\mathbf{q}(t_1,m_2)]\int_{-\infty}^{\infty} d^Nq(n_2,m_1)\int_{\mathbf{q}(n_2,m_1)}^{\mathbf{q}(n_2,m_2)}\mathscr{D}[\mathbf{q}(n_2,t_2)]\int_{-\infty}^{\infty} d^Nq(n_2,m_2)\int_{\mathbf{q}(n_2,m_2)}^{\mathbf{q}(\mathcal{N},m_2)}\mathscr{D}[\mathbf{q}(t_1,m_2)]\nonumber
        \\&\times\int_{-\infty}^{\infty} d^Nq(\mathcal{N},m_2)\int_{\mathbf{q}(\mathcal{N},m_2)}^{\mathbf{q}(\mathcal{N},\mathcal{N})}\mathscr{D}[\mathbf{q}(\mathcal{N},t_2)]\Bigg)e^{\frac{i}{\hbar}\int_{\{\Gamma:\Gamma\in\bigcup_{l=1}^5\mathscr{B}_l\}}\mathscr{L}}\;.\label{5first propagator}
    \end{align}
    The fourth path, see figure \ref{fig:formatD}, is nothing but the fifth one, see figure \ref{fig:formatE}, in case of $m_2=\mathcal{N}$ for every single $m_1$ together with $m_1=\mathcal{N}-1$ and $m_2=\mathcal{N}$. Moreover, the fifth path can be reduced to be the third path in the case of $m_2=m_1$ for every single $m_1$ and the second one in the case of $m_1=m_2=\mathcal{N}$. For the first path can be obtained from the fifth path by letting $m_1=m_2=0$, but the case of $m_1=0$ with any $m_2\neq0$ is the same with the case of $m_2=m_1$(third path). This over-counted problem could be settled by fixing $n_1=1$, see appendix \ref{AppendixC}. The propagator, therefore, can be further simplified to 
    \begin{align}
         K^{(5)} =& \sum_{n_2=2}^{\mathcal{N}-1}\sum_{m_2=m_1}^{\mathcal{N}}\sum_{m_1=0}^{\mathcal{N}}\int_{\mathbf{q}(0,0)}^{\mathbf{q}(1,0)}\mathscr{D}[\mathbf{q}(t_1,0)]\int_{-\infty}^{\infty} d^Nq(1,0)\int_{\mathbf{q}(1,0)}^{\mathbf{q}(1,m_1)}\mathscr{D}[\mathbf{q}(1,t_2)]\int_{-\infty}^{\infty} d^Nq(1,m_1)\nonumber
         \\&\times\int_{\mathbf{q}(1,m_1)}^{\mathbf{q}(n_2,m_1)}\mathscr{D}[\mathbf{q}(t_1,m_1)]\int_{-\infty}^{\infty} d^Nq(n_2,m_1)\int_{\mathbf{q}(n_2,m_1)}^{\mathbf{q}(n_2,m_2)}\mathscr{D}[\mathbf{q}(n_2,t_2)]\int_{-\infty}^{\infty} d^Nq(n_2,m_2)\nonumber
         \\&\times\int_{\mathbf{q}(n_2,m_2)}^{\mathbf{q}(\mathcal{N},m_2)}\mathscr{D}[\mathbf{q}(t_1,m_2)]\int_{-\infty}^{\infty} d^Nq(\mathcal{N},m_2)\int_{\mathbf{q}(\mathcal{N},m_2)}^{\mathbf{q}(\mathcal{N},\mathcal{N})}\mathscr{D}[\mathbf{q}(\mathcal{N},t_2)]e^{\frac{i}{\hbar}\int_{\{\Gamma:\Gamma\in\bigcup_{l=1}^5\mathscr{B}_l\}}\mathscr{L}}\;.\label{5first propagator complete}
    \end{align}
    Here the order of summation is still crucial since we would have the repeated paths where $n_2$ is arbitrary at $m_2=m_1$. Diagrammatically, what we do in \eqref{5first propagator complete} is first shifting the horizontal line from $(1,0)-(n_2,0)$ to $(1,\mathcal{N})-(n_2,\mathcal{N})$ and second shifting the another horizontal line from $(n_2,0)-(\mathcal{N},0)$ to $(n_2,\mathcal{N})-(\mathcal{N},\mathcal{N})$ , where the second horizontal line must be always above the first one, for $n_2=[2,\mathcal{N}]$.
    \\
    \\
    By proceeding this deformation of the curve, we could account for all configurations of the path. The propagator, including the normalising factor\footnote{This normalising factors are different from the normalising factor in equation \eqref{path functional}, but they are specified to keep a preserved norm of quantum states.}, could be expressed as
	\begin{align}
	    K^{(\text{All})}=&\sum_{m_{\mathcal{N}-1}\ge\cdots\ge m_2\ge m_1\ge0}^\mathcal{N}\mathscr{N}_{m_I}\int_{\mathbf{q}(0,0)}^{\mathbf{q}(1,0)}\mathscr{D}[\mathbf{q}(t_1,0)]\bigg(\prod_{i=1}^{\mathcal{N}-1}\int_{-\infty}^{\infty} d^Nq(i,m_{i-1})\int_{\mathbf{q}(i,m_{i-1})}^{\mathbf{q}(i,m_i)}\mathscr{D}[\mathbf{q}(i,t_2)]\int_{-\infty}^{\infty} d^Nq(i,m_{i})\nonumber
	    \\&\times\int_{\mathbf{q}(i,m_{i})}^{\mathbf{q}(i+1,m_{i})}\mathscr{D}[\mathbf{q}(t_1,m_{i+1})]\bigg)\int_{-\infty}^{\infty} d^Nq(\mathcal{N},m_{\mathcal{N}-1})\int_{\mathbf{q}(\mathcal{N},m_{\mathcal{N}-1})}^{\mathbf{q}(\mathcal{N},\mathcal{N})}\mathscr{D}[\mathbf{q}(\mathcal{\mathcal{N}},t_2)]e^{\frac{i}{\hbar}\int_{\{\Gamma:\Gamma\in\bigcup_{l=1}^{2\mathcal{N}-1}\mathscr{B}_l\}}\mathscr{L}}\;, \label{all t1 propagator}
	\end{align}
	where $m_I=m_1m_2\cdots m_{\mathcal{N}-1}$ is multi-indices and $m_{0}\equiv 0$.
	Alternatively, in terms of diagram, we could imagine that we slice the horizontal line as $\mathcal{N}$ pieces($n_l=1\;;\; l=0,1,2,\dots,\mathcal{N}-1$) then shift the second piece onward at $m_1, m_2,\dots,m_{\mathcal{N}-1}$, respectively, with conditions: $m_i\geqslant m_j$, where $i>j$.
	\\
	\\
	We know that the equation \eqref{all t1 propagator} is just only the propagator starting in the $t_1$-direction and all possible types of deformation. However, there is a path starting in the $t_2$-direction as well. We therefore employ the symmetry of the lattice structure under interchange $t_1$ and $t_2$. Hence, the propagator can be presented as follows
	\begin{equation}
	   K(\mathbf{q}(\mathbf{t}(s'')),s''; \mathbf{q}(\mathbf{t}(s')), s') =\int_{\mathbf{q}(\mathbf{t}(s'))}^{\mathbf{q}(\mathbf{t}(s''))}\mathbb{D}[\mathbf{q}(\mathbf{t}(s));\Gamma\in\mathscr{B}]e^{\frac{i}{\hbar}\int_{\{\Gamma:\Gamma\in\mathscr{B}\}}\mathscr{L}}\;,\label{KK1}
	\end{equation}
	where $\mathscr{B}$ is a family of paths connecting between $\mathbf{t}(s')$ and $\mathbf{t}(s'')$ on the space of 2 time variables and
	\begin{align}
	    \int_{\mathbf{q}(s')}^{\mathbf{q}(s'')}\mathbb{D}[\mathbf{q}(s);\Gamma\in\mathscr{B}]=&\int_{\mathbf{q}(\mathbf{t}(s'))}^{\mathbf{q}(\mathbf{t}(s''))}\mathbb{D}[\mathbf{q}(\mathbf{t}(s));\Gamma\in\mathscr{B}]\nonumber
	    \\=&\lim_{\substack{\mathcal{N}\to\infty\\\epsilon_{1,2}\to 0}}\Bigg\{\sum_{m_{\mathcal{N}-1}\ge\cdots\ge m_2\ge m_1\ge0}^\mathcal{N}\mathscr{N}_{m_I}\int_{\mathbf{q}(0,0)}^{\mathbf{q}(\epsilon_1,0)}\mathscr{D}[\mathbf{q}(t_1,0)]\bigg(\prod_{i=1}^{\mathcal{N}-1}\int_{-\infty}^{\infty} d^Nq(i\epsilon_1,m_{i-1}\epsilon_2)\nonumber
	    \\&\times\int_{\mathbf{q}(i\epsilon_1,m_{i-1}\epsilon_2)}^{\mathbf{q}(i\epsilon_1,m_i\epsilon_2)}\mathscr{D}[\mathbf{q}(i\epsilon_1,t_2)]\int_{-\infty}^{\infty} d^Nq(i\epsilon_1,m_{i}\epsilon_2)\int_{\mathbf{q}(i\epsilon_1,m_{i}\epsilon_2)}^{\mathbf{q}((i+1)\epsilon_1,m_{i}\epsilon_2)}\mathscr{D}[\mathbf{q}(t_1,m_{i}\epsilon_2)]\bigg)\nonumber
	    \\&\times\int_{-\infty}^{\infty} d^Nq(\mathcal{N}\epsilon_1,m_{\mathcal{N}-1}\epsilon_2)\int_{\mathbf{q}(\mathcal{N}\epsilon_1,m_{\mathcal{N}-1}\epsilon_2)}^{\mathbf{q}(\mathcal{N}\epsilon_1,\mathcal{N}\epsilon_2)}\mathscr{D}[\mathbf{q}(\mathcal{N}\epsilon_1,t_2)]\nonumber
	    \\&+\sum_{n_{\mathcal{N}-1}\ge\cdots\ge n_2\ge n_1\ge0}^\mathcal{N}\mathscr{N}_{n_I}\int_{\mathbf{q}(0,0)}^{\mathbf{q}(0,\epsilon_2)}\mathscr{D}[\mathbf{q}(0,t_2)]\bigg(\prod_{i=1}^{\mathcal{N}-1}\int_{-\infty}^{\infty} d^Nq(n_{i-1}\epsilon_1,i\epsilon_2)\nonumber
	    \\&\times\int_{\mathbf{q}(n_{i-1}\epsilon_1,i\epsilon_2)}^{\mathbf{q}(n_i\epsilon_1,i\epsilon_2)}\mathscr{D}[\mathbf{q}(t_1,i\epsilon_2)]\int_{-\infty}^{\infty} d^Nq(n_{i}\epsilon_1,i\epsilon_2)\int_{\mathbf{q}(n_{i}\epsilon_1,i\epsilon_2)}^{\mathbf{q}(n_{i}\epsilon_1,(i+1)\epsilon_2)}\mathscr{D}[\mathbf{q}(n_{i}\epsilon_1,t_2)]\bigg)\nonumber
	    \\&\times\int_{-\infty}^{\infty} d^Nq(n_{\mathcal{N}-1}\epsilon_1,\mathcal{N}\epsilon_2)\int_{\mathbf{q}(n_{\mathcal{N}-1}\epsilon_1,\mathcal{N}\epsilon_2)}^{\mathbf{q}(\mathcal{N}\epsilon_1,\mathcal{N}\epsilon_2)}\mathscr{D}[\mathbf{q}(t_1,\mathcal{N}\epsilon_2)]\Bigg\}\;.\label{the new 2D gauage of propagator}
	\end{align}
	Here, $\epsilon_1$ and $\epsilon_2$ are put back into the formula. The new notation $\int\mathbb{D}[\mathbf{q}(\mathbf{t}(s));\Gamma\in\mathscr{B}]$ is a extended definition of the standard one $\int\mathscr{D}[\mathbf{q}(t)]$. Then the propagator in \eqref{KK1} represents sum all possible paths not only on the configuration space(standard one), but also all possible paths $\Gamma$ on the time space, see figure \ref{possible paths on multi-time space}.
    \begin{figure}[h]
        \centering
            \tdplotsetmaincoords{60}{125}
            \begin{tikzpicture}
		        [tdplot_main_coords,
		    	cube/.style={dotted, black},
		    	grid/.style={very thin,gray},
		    	axis/.style={->,black,thick}]


        	\draw[axis] (0,0,0) -- (7,0,0) node[anchor=north]{$t_1$};
          	\draw[axis] (0,0,0) -- (0,7,0) node[anchor=west]{$t_2$};
        	\draw[axis] (0,0,0) -- (0,0,7) node[anchor=west]{$\mathbf{q}$};

        	\draw[cube, densely dotted, line width = 1 pt] (0,0,0) -- (0,6,0) -- (6,6,0) -- (6,0,0) -- cycle;
        	\draw[cube] (0,0,2) -- (0,6,2) -- (6,6,2) -- (6,0,2) -- cycle;
        	\draw[cube] (0,0,4) -- (0,6,4) -- (6,6,4) -- (6,0,4) -- cycle;
        	\draw[cube] (0,0,1) -- (0,6,1) -- (6,6,1) -- (6,0,1) -- cycle;
        	\draw[cube] (0,0,3) -- (0,6,3) -- (6,6,3) -- (6,0,3) -- cycle;
        	\draw[cube] (0,0,5) -- (0,6,5) -- (6,6,5) -- (6,0,5) -- cycle;
    	    \draw[cube, densely dotted, line width = 1 pt] (0,0,6) -- (0,6,6) -- (6,6,6) -- (6,0,6) -- cycle;
	        \draw[cube] (0,2,0) -- (6,2,0);
	        \draw[cube] (0,4,0) -- (6,4,0);
	        \draw[cube] (0,1,0) -- (6,1,0);
	        \draw[cube] (0,3,0) -- (6,3,0);
	        \draw[cube] (0,5,0) -- (6,5,0);
    	    \draw[cube] (0,2,6) -- (6,2,6);
    	    \draw[cube] (0,4,6) -- (6,4,6);
    	    \draw[cube] (0,1,6) -- (6,1,6);
    	    \draw[cube] (0,3,6) -- (6,3,6);
    	    \draw[cube] (0,5,6) -- (6,5,6);
    	    \draw[cube] (2,0,0) -- (2,6,0);
    	    \draw[cube] (4,0,0) -- (4,6,0);
    	    \draw[cube] (1,0,0) -- (1,6,0);
    	    \draw[cube] (3,0,0) -- (3,6,0);
    	    \draw[cube] (5,0,0) -- (5,6,0);
    	    \draw[cube] (2,0,6) -- (2,6,6);
	        \draw[cube] (4,0,6) -- (4,6,6);
    	    \draw[cube] (1,0,6) -- (1,6,6);
	        \draw[cube] (3,0,6) -- (3,6,6);
    	    \draw[cube] (5,0,6) -- (5,6,6);
	
	        \draw[cube, densely dotted, line width = 1 pt] (0,0,0) -- (0,0,6);
	        \draw[cube] (0,2,0) -- (0,2,6);
	        \draw[cube] (0,1,0) -- (0,1,6);
	        \draw[cube] (0,3,0) -- (0,3,6);
	        \draw[cube] (0,5,0) -- (0,5,6);
	        \draw[cube] (2,0,0) -- (2,0,6);
	        \draw[cube] (1,0,0) -- (1,0,6);
	        \draw[cube] (3,0,0) -- (3,0,6);
	        \draw[cube] (5,0,0) -- (5,0,6);
	        \draw[cube] (0,4,0) -- (0,4,6);
	        \draw[cube] (4,0,0) -- (4,0,6);
	        \draw[cube] (6,2,0) -- (6,2,6);
	        \draw[cube] (6,4,0) -- (6,4,6);
	        \draw[cube] (6,1,0) -- (6,1,6);
	        \draw[cube] (6,3,0) -- (6,3,6);
	        \draw[cube] (6,5,0) -- (6,5,6);
	        \draw[cube] (2,6,0) -- (2,6,6);
	        \draw[cube] (4,6,0) -- (4,6,6);
	        \draw[cube] (1,6,0) -- (1,6,6);
	        \draw[cube] (3,6,0) -- (3,6,6);
	        \draw[cube] (5,6,0) -- (5,6,6);
	        \draw[cube, densely dotted, line width = 1 pt] (0,6,0) -- (0,6,6);
            \draw[cube, densely dotted, line width = 1 pt] (6,0,0) -- (6,0,6);
	        \draw[cube, densely dotted, line width = 1 pt] (6,6,0) -- (6,6,6);
	        
	        \draw[-,line width = 1pt, postaction ={on each segment = {mid arrow}}] (1,1,1) to [bend left = 70] (1,6,5) -- (4,6,5);
	        \draw[-,line width = 1pt, postaction ={on each segment = {mid arrow}}] (1,1,1) -- (3,1,1) -- (3,3,1) -- (4,3,1) to [bend right = 50] (4,5,1.5) to [bend right = 20] (4,6,5);
	        
	        \draw[densely dashed, line width = 1pt, postaction ={on each segment = {mid arrow}}] (1,1,0) -- (1,1,0) -- (1,3,0) -- (1,3,0) -- (1,6,0) -- (4,6,0);
	        \draw[densely dashed, line width = 1pt, postaction ={on each segment = {mid arrow}}] (1,1,0) -- (3,1,0) -- (3,3,0) -- (4,3,0) -- (4,6,0) -- (4,6,0);
	        
	        \draw[loosely dashed, line width = 0.5pt] (1,1,0) -- (1,1,1);
	        \draw[loosely dashed, line width = 0.5pt] (4,6,0) -- (4,6,5);
	        
	        \node[circle, fill, inner sep=1.5 pt] at (1,1,1) {};
	        \node[circle, fill, inner sep=1.5 pt] at (4,6,5) {};
	        \node[circle, fill, inner sep=1.25 pt] at (1,1,0) {};
	        \node[circle, fill, inner sep=1.25 pt] at (4,6,0) {};
	        
	        \node at (1,0,1) {$\mathbf{q}(t_1',t_2')$};
	        \node at (4,7,5) {$\mathbf{q}(t_1'',t_2'')$};
	        \node at (0.5,1.5,0) {$(t_1',t_2')$};
	        \node at (3.5,6.5,0) {$(t_1'',t_2'')$};
	        \node at (4.5,3,0) {$\Gamma$};
        \end{tikzpicture}
        \caption{All possible paths on the 2-time space}\label{possible paths on multi-time space}
    \end{figure}
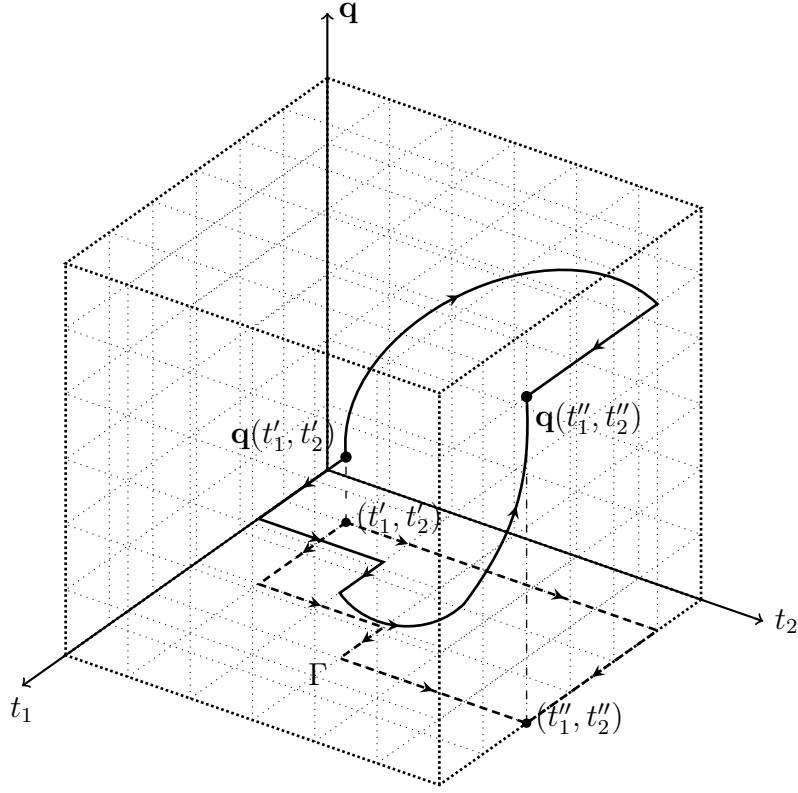
    \\
    \\
    What we obtain in \eqref{KK1} is just the case of two times. Therefore, the process can be directly extended to the case of arbitrary $N$ times by using a diagrammatic method. Here, we will sketch the idea in the case of three times. The propagator in this case will get a contribution from all possible types of deformation starting in $t_1$-direction, see figure \ref{3D propagator}. Mathematically, we just simply shift each step in the $t_1$-direction, resulting in
    \begin{align}
	    &\int_{\mathbf{q}(s')}^{\mathbf{q}(s'')}\mathbb{D}[\mathbf{q}(s);\Gamma\in\mathscr{B}] = \int_{\mathbf{q}(\mathbf{t}(s'))}^{\mathbf{q}(\mathbf{t}(s''))}\mathbb{D}[\mathbf{q}(\mathbf{t}(s));\Gamma\in\mathscr{B}] \nonumber
	    \\&=\lim_{\substack{\mathcal{N}\to\infty\\\epsilon_{1,2,3}\to 0}}\Bigg(\sum_{l_{\mathcal{N}-1}\ge\cdots\ge l_2\ge l_1\ge0}^\mathcal{N}\sum_{m_{\mathcal{N}-1}\ge\cdots\ge m_2\ge m_1\ge0}^\mathcal{N}\mathscr{N}_{\Gamma}\int_{\mathbf{q}(0,0,0)}^{\mathbf{q}(\epsilon_1,0,0)}\mathscr{D}[\mathbf{q}(t_1,0,0)]\int_{-\infty}^{\infty} d^Nq(\epsilon_1,0,0) \nonumber
	    \\&\;\times\Bigg\{\int_{\mathbf{q}(\epsilon_1,0,0)}^{\mathbf{q}(\epsilon_1,m_1\epsilon_2,0)}\mathscr{D}[\mathbf{q}(\epsilon_1,t_2,0)]\int_{-\infty}^{\infty} d^Nq(\epsilon_1,m_1\epsilon_2,0)\int_{\mathbf{q}(\epsilon_1,m_1\epsilon_2,0)}^{\mathbf{q}(\epsilon_1,m_1\epsilon_2,l_1\epsilon_3)}\mathscr{D}[\mathbf{q}(\epsilon_1,m_1\epsilon_2,t_3)] \nonumber
	    \\&\;+\int_{\mathbf{q}(\epsilon_1,0,0)}^{\mathbf{q}(\epsilon_1,0,l_1\epsilon_3)}\mathscr{D}[\mathbf{q}(\epsilon_1,0,t_3)]\int_{-\infty}^{\infty} d^Nq(\epsilon_1,0,l_1\epsilon_3)\int_{\mathbf{q}(\epsilon_1,0,l_1\epsilon_3)}^{\mathbf{q}(\epsilon_1,m_1\epsilon_2,l_1\epsilon_3)}\mathscr{D}[\mathbf{q}(\epsilon_1,t_2,l_1\epsilon_3)]\Bigg\}\nonumber
	    \\&\;\times\int_{-\infty}^{\infty} d^Nq(\epsilon_1,m_1\epsilon_2,l_1\epsilon_3)\int_{\mathbf{q}(\epsilon_1,m_1\epsilon_2,l_1\epsilon_3)}^{\mathbf{q}(2\epsilon_1,m_1\epsilon_2,l_1\epsilon_3)}\mathscr{D}[\mathbf{q}(t_1,m_1\epsilon_2,l_1\epsilon_3)]\int_{-\infty}^{\infty} d^Nq(2\epsilon_1,m_1\epsilon_2,l_1\epsilon_3)\nonumber
	    \\&\;\times\Bigg\{\int_{\mathbf{q}(2\epsilon_1,m_1\epsilon_2,l_1\epsilon_3)}^{\mathbf{q}(2\epsilon_1,m_2\epsilon_2,l_1\epsilon_3)}\mathscr{D}[\mathbf{q}(2\epsilon_1,t_2,l_1\epsilon_3)]\int_{-\infty}^{\infty} d^Nq(2\epsilon_1,m_2\epsilon_2,l_1\epsilon_3)\int_{\mathbf{q}(2\epsilon_1,m_2\epsilon_1,l_1\epsilon_1)}^{\mathbf{q}(2\epsilon_1,m_2\epsilon_2,l_2\epsilon_3)}\mathscr{D}[\mathbf{q}(2\epsilon_1,m_2\epsilon_2,t_3)] \nonumber
	    \\&\;+\int_{\mathbf{q}(2\epsilon_1,m_1\epsilon_2,l_1\epsilon_3)}^{\mathbf{q}(2\epsilon_1,m_1\epsilon_2,l_2\epsilon_3)}\mathscr{D}[\mathbf{q}(2\epsilon_1,m_1\epsilon_2,t_3)]\int_{-\infty}^{\infty} d^Nq(2\epsilon_1,m_1\epsilon_2,l_2\epsilon_3)\int_{\mathbf{q}(2\epsilon_1,m_1\epsilon_2,l_2\epsilon_3)}^{\mathbf{q}(2\epsilon_1,m_2\epsilon_2,l_2\epsilon_3)}\mathscr{D}[\mathbf{q}(2\epsilon_1,t_2,l_2\epsilon_3)]\Bigg\}\nonumber
	    \\&\;\cdots\int_{\mathbf{q}((\mathcal{N}-1)\epsilon_1,m_{\mathcal{N}-1}\epsilon_2,l_{\mathcal{N}-1}\epsilon_3)}^{\mathbf{q}(\mathcal{N}\epsilon_1,m_{\mathcal{N}-1}\epsilon_2,l_{\mathcal{N}-1}\epsilon_3)}\mathscr{D}[\mathbf{q}(t_1,m_{\mathcal{N}-1}\epsilon_2,l_{\mathcal{N}-1}\epsilon_3)]\int_{-\infty}^{\infty} d^Nq(\mathcal{N}\epsilon_1,m_{\mathcal{N}-1}\epsilon_2,l_{\mathcal{N}-1}\epsilon_3) \nonumber
	    \\&\;\times\Bigg\{\int_{\mathbf{q}(\mathcal{N}\epsilon_1,m_{\mathcal{N}-1}\epsilon_2,l_{\mathcal{N}-1}\epsilon_3)}^{\mathbf{q}(\mathcal{N}\epsilon_1,\mathcal{N}\epsilon_2,l_{\mathcal{N}-1}\epsilon_3)}\mathscr{D}[\mathbf{q}(\mathcal{N}\epsilon_1,t_2,l_{\mathcal{N}-1}\epsilon_3)]\int_{-\infty}^{\infty} d^Nq(\mathcal{N}\epsilon_1,\mathcal{N}\epsilon_2,l_{\mathcal{N}-1}\epsilon_3)\int_{\mathbf{q}(\mathcal{N}\epsilon_1,\mathcal{N}\epsilon_2,l_{\mathcal{N}-1}\epsilon_3)}^{\mathbf{q}(\mathcal{N}\epsilon_1,\mathcal{N}\epsilon_2,\mathcal{N}\epsilon_3)}\mathscr{D}[\mathbf{q}(\mathcal{N}\epsilon_1,\mathcal{N}\epsilon_2,t_3)] \nonumber
	    \\&\;+\int_{\mathbf{q}(\mathcal{N}\epsilon_1,m_{\mathcal{N}-1}\epsilon_2,l_{\mathcal{N}-1}\epsilon_3)}^{\mathbf{q}(\mathcal{N}\epsilon_1,m_{\mathcal{N}-1}\epsilon_2,\mathcal{N}\epsilon_3)}\mathscr{D}[\mathbf{q}(\mathcal{N}\epsilon_1,m_{\mathcal{N}-1}\epsilon_2,t_3)]\int_{-\infty}^{\infty} d^Nq(\mathcal{N}\epsilon_1,m_{\mathcal{N}-1}\epsilon_2,\mathcal{N}\epsilon_3)\int_{\mathbf{q}(\mathcal{N}\epsilon_1,m_{\mathcal{N}-1}\epsilon_2,\mathcal{N}\epsilon_3)}^{\mathbf{q}(\mathcal{N}\epsilon_1,\mathcal{N}\epsilon_2,\mathcal{N}\epsilon_3)}\mathscr{D}[\mathbf{q}(\mathcal{N}\epsilon_1,t_2,\mathcal{N}\epsilon_3)]\Bigg\}\nonumber
	    \\&\; + (\text{the $t_2$-symmetric term}) + (\text{the $t_3$-symmetric term})\Bigg)\;,\label{the new 3D gauage of propagator}
	\end{align}
	where $\epsilon_3$ is a width of step-evolution in $t_3$-direction and $\mathscr{B}$ is a family of paths connecting between $\mathbf{t}(s')$ and $\mathbf{t}(s'')$ on the space of 3 time variables.
	\begin{figure}[h]
	    \centering
	    \begin{tikzpicture}[scale=0.9]
	    
	         \path[->, draw = black]
	                (0,0) to (7,0);
	                
	         \path[->, draw = black]
	                (0,0) to (0,3);
	                
	         \path[->, draw = black]
	                (0,0) to (-2,-2);
	                
	         \path[draw = black,line width = 1.25 pt, postaction ={on each segment = {mid arrow}}]
	                (0,0) to (2.5,0)
	                (1.5,0.5) to (4,0.5)
	                (3,1) to (5.5,1);
	                
	         \path[dashed, draw = black,line width = 1 pt, postaction ={on each segment = {mid arrow}}]
	                (2.5,0) to (1.5,-1)
	                (1.5,-1) to (1.5,0.5)
	                (4,0.5) to (3,-0.5)
	                (3,-0.5) to (3,1);
	                
	         \path[densely dotted, draw = black, line width = 1 pt, postaction ={on each segment = {mid arrow}}]
	                (2.5,0) to (2.5,1.5)
	                (2.5,1.5) to (1.5,0.5)
	                (4,0.5) to (4,2)
	                (4,2) to (3,1);
 	          
 	          \node[circle, fill, inner sep=1.25 pt] at (0,0) {};
 	          \node[circle, fill, inner sep=1.25 pt] at (2.5,0) {};
 	          \node[circle, fill, inner sep=1.25 pt] at (1.5,0.5) {};
 	          \node[circle, fill, inner sep=1.25 pt] at (4,0.5) {};
 	          \node[circle, fill, inner sep=1.25 pt] at (3,1) {};
 	          \node[circle, fill, inner sep=1.25 pt] at (5.5,1) {};
 	          
 	          \node at (0,0.25) {\scriptsize{$(0,0,0)$}};
 			  \node at (2.5,-0.25) {\scriptsize{$(1,0,0)$}};
 			  \node at (1.5,0.75) {\scriptsize{$(1,m_1,l_1)$}}; 
 			  \node at (4,0.25) {\scriptsize{$(2,m_1,l_1)$}};
 			  \node at (3,1.25) {\scriptsize{$(2,m_2,l_2)$}};
 			  \node at (5.5,1.25) {\scriptsize{$(3,m_2,l_2)$}}; 
 			  
 			  \node at (0,3.25) {$t_2$};
 			  \node at (7.25,0) {$t_1$};
 			  \node at (-2.25,-2.25) {$t_3$}; 
	    \end{tikzpicture}
	    \caption{The deformation of paths starting in $t_1$-direction on the three-dimensional space of time variables.}
	    \label{3D propagator}
	\end{figure}
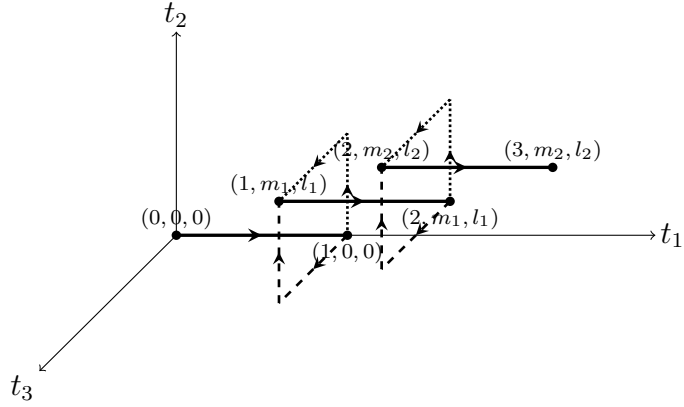
	\\
	Here, we define all possible permutations $\mathcal{P}$ of the functional measure such that
	\begin{align}
	    &\int_{\mathbf{q}(n\epsilon_1,m_i\epsilon_2,l_i\epsilon_3)}^{\mathbf{q}(n\epsilon_1,m_{i+1}\epsilon_2,l_{i+1}\epsilon_3)}\mathcal{P}\bigg[\mathscr{D}[\mathbf{q}(t_2)],\mathscr{D}[\mathbf{q}(t_3)]\bigg] \nonumber
	    \\&= \frac{1}{\mathcal{P}(2,r_{i+1})}\Bigg(\int_{\mathbf{q}(n\epsilon_1,m_i\epsilon_2,l_i\epsilon_3)}^{\mathbf{q}(n\epsilon_1,m_{i+1}\epsilon_2,l_i\epsilon_3)}\mathscr{D}[\mathbf{q}(n\epsilon_1,t_2,l_i\epsilon_3)]\int_{-\infty}^{\infty} d^Nq(n\epsilon_1,m_{i+1}\epsilon_2,l_i\epsilon_3)\nonumber
	    \\&\;\;\;\times\int_{\mathbf{q}(n\epsilon_1,m_{i+1}\epsilon_2,l_i\epsilon_3)}^{\mathbf{q}(n\epsilon_1,m_{i+1}\epsilon_2,l_{i+1}\epsilon_3)}\mathscr{D}[\mathbf{q}(n\epsilon_1,m_{i+1}\epsilon_2,t_3)] \nonumber
	    \\&\;\;\;+\int_{\mathbf{q}(n\epsilon_1,m_i\epsilon_2,l_i\epsilon_3)}^{\mathbf{q}(n\epsilon_1,m_i\epsilon_2,l_{i+1}\epsilon_3)}\mathscr{D}[\mathbf{q}(n\epsilon_1,m_i\epsilon_2,t_3)]\int_{-\infty}^{\infty} d^Nq(n\epsilon_1,m_i\epsilon_2,l_{i+1}\epsilon_3)\int_{\mathbf{q}(n\epsilon_1,m_i\epsilon_2,l_{i+1}\epsilon_3)}^{\mathbf{q}(n\epsilon_1,m_{i+1}\epsilon_2,l_{i+1}\epsilon_3)}\mathscr{D}[\mathbf{q}(n\epsilon_1,t_2,l_{i+1}\epsilon_3)]\Bigg)\;,\nonumber\\\label{anticommute}
	\end{align}
	where $\mathbf{q}(t_i)$ means that $t_i$ is active but $t_{j\neq i}$ are fixed. The factor $\mathcal{P}(2,r_{i+1})=\frac{2!}{(2-r_{i+1})!}$, where $\;r_{i+1}=\delta_{m_i,m_{i+1}}+\delta_{l_i,l_{i+1}}$ is placed to avoid the redundant path in some situations. Let us illustrate a simple case as follows. If $t_1$ is fixed and $t_2$ or $t_3$ does not activate, the redundancy will exist. With condition $n_0=m_0=l_0=0$, hence equation \eqref{the new 3D gauage of propagator} will be simply expressed as 
	\begin{eqnarray}
	    \int_{\mathbf{q}(\mathbf{t}(s'))}^{\mathbf{q}(\mathbf{t}(s''))}\mathbb{D}[\mathbf{q}(\mathbf{t}(s));\Gamma\in\mathscr{B}]&=&\lim_{\substack{\mathcal{N}\to\infty\\\epsilon_{1,2,3}\to 0}}\Bigg(\sum_{l_{\mathcal{N}-1}\ge\cdots\ge l_2\ge l_1\ge0}^\mathcal{N}\sum_{m_{\mathcal{N}-1}\ge\cdots\ge m_2\ge m_1\ge0}^\mathcal{N}\mathscr{N}_{\Gamma}
	    \int_{\mathbf{q}(0,0,0)}^{\mathbf{q}(\epsilon_1,0,0)}\mathscr{D}[\mathbf{q}(t_1,0,0)]\nonumber
	    \\&&\times\bigg(\prod_{i=1}^{\mathcal{N}-1}\int_{-\infty}^{\infty} d^Nq(i\epsilon_1,m_{i-1}\epsilon_2,l_{i-1}\epsilon_3)\int_{\mathbf{q}(i\epsilon_1,m_{i-1}\epsilon_2,l_{i-1}\epsilon_3)}^{\mathbf{q}(i\epsilon_1,m_i\epsilon_2,l_i\epsilon_3)}\mathcal{P}\Bigg[\mathscr{D}[\mathbf{q}(t_2)],\mathscr{D}[\mathbf{q}(t_3)]\Bigg]\nonumber
	    \\&&\times\int_{-\infty}^{\infty} d^Nq(i\epsilon_1,m_i\epsilon_2,l_i\epsilon_3)\int_{\mathbf{q}(i\epsilon_1,m_i\epsilon_2,l_i\epsilon_3)}^{\mathbf{q}((i+1)\epsilon_1,m_i\epsilon_2,l_i\epsilon_3)}\mathscr{D}[\mathbf{q}(t_1,m_i\epsilon_2,l_i\epsilon_3)]\bigg)\nonumber
	    \\&&\times\int_{-\infty}^{\infty} d^Nq(\mathcal{N}\epsilon_1,m_{\mathcal{N}-1}\epsilon_2,l_{\mathcal{N}-1}\epsilon_3)\int_{\mathbf{q}(\mathcal{N}\epsilon_1,m_{\mathcal{N}-1}\epsilon_2,l_{\mathcal{N}-1}\epsilon_3)}^{\mathbf{q}(\mathcal{N}\epsilon_1,\mathcal{N}\epsilon_2,\mathcal{N}\epsilon_3)}\mathcal{P}\Bigg[\mathscr{D}[\mathbf{q}(t_2)],\mathscr{D}[\mathbf{q}(t_3)]\Bigg]\nonumber
	    \\&&+ (\text{the $t_2$-symmetric term}) + (\text{the $t_3$-symmetric term})\Bigg)\;.\label{the new 3D gauage of propagator complete}
	\end{eqnarray}
	\\
	\\
	For $N$-dimensional of time space, the functional measure over all possible spatial-temporal paths could be presented as
	\begin{eqnarray}
	     &&\int_{\mathbf{q}(s')}^{\mathbf{q}(s'')}\mathbb{D}[\mathbf{q}(s);\Gamma\in\mathscr{B}]=\int_{\mathbf{q}(\mathbf{t}(s'))}^{\mathbf{q}(\mathbf{t}(s''))}\mathbb{D}[\mathbf{q}(\mathbf{t}(s));\Gamma\in\mathscr{B}]\nonumber
	     \\&&=\lim_{\substack{\mathcal{N}\to\infty\\\epsilon_{1,2,\dots,N}\to 0}}\Bigg(\sum_{\alpha^N_{\mathcal{N}-1}\ge\cdots\alpha^N_2\ge\alpha^N_1\ge0}^\mathcal{N}\cdots\sum_{\alpha^3_{\mathcal{N}-1}\ge\cdots\alpha^3_2\ge\alpha^3_1\ge0}^\mathcal{N}\sum_{\alpha^2_{\mathcal{N}-1}\ge\cdots\alpha^2_2\ge\alpha^2_1\ge0}^\mathcal{N}\mathscr{N}_{\Gamma}\int_{\mathbf{q}(0,0,\dots,0)}^{\mathbf{q}(\epsilon_1,0,\dots,0)}\mathscr{D}[\mathbf{q}(t_1)]\nonumber
	     \\&&\;\times\Bigg(\prod_{j=1}^{\mathcal{N}-1}\int_{-\infty}^{\infty} d^Nq(j\epsilon_1,\alpha_{j-1}^2\epsilon_2,\dots,\alpha_{j-1}^N\epsilon_N) \int_{\mathbf{q}(j\epsilon_1,\alpha^2_{j-1}\epsilon_2,\dots\alpha^N_{j-1}\epsilon_N)}^{\mathbf{q}(j\epsilon_1,\alpha^2_j\epsilon_2,\dots\alpha^N_j\epsilon_N)}\mathcal{P}\Bigg[\mathscr{D}[\mathbf{q}(t_2)],\mathscr{D}[\mathbf{q}(t_3)],\cdots,\mathscr{D}[\mathbf{q}(t_N)]\Bigg] \nonumber
	     \\&&\;\times\int_{-\infty}^{\infty} d^Nq(j\epsilon_1,\alpha_j^2\epsilon_2,\dots,\alpha_j^N\epsilon_N) \int_{\mathbf{q}(j\epsilon_1,\alpha^2_j\epsilon_2,\dots\alpha^N_j\epsilon_N)}^{\mathbf{q}((j+1)\epsilon_1,\alpha^2_j\epsilon_2,\dots\alpha^N_j\epsilon_N)}\mathscr{D}[\mathbf{q}(t_1)]\Bigg) \nonumber
	     \\&&\;\times\int_{-\infty}^{\infty} d^Nq(\mathcal{N}\epsilon_1,\alpha^2_{\mathcal{N}-1}\epsilon_2,\dots,\alpha^N_{\mathcal{N}-1}\epsilon_N) \int_{\mathbf{q}(\mathcal{N}\epsilon_1,\alpha^2_{\mathcal{N}-1}\epsilon_2,\dots,\alpha^N_{\mathcal{N}-1}\epsilon_N)}^{\mathbf{q}(\mathcal{N}\epsilon_1,\mathcal{N}\epsilon_2,\dots,\mathcal{N}\epsilon_N)}\mathcal{P}\Bigg[\mathscr{D}[\mathbf{q}(t_2)],\mathscr{D}[\mathbf{q}(t_3)],\cdots,\mathscr{D}[\mathbf{q}(t_N)]\Bigg]\nonumber
	     \\&&\;+(\text{all symmetric terms})\Bigg)\;,
	\end{eqnarray}
	where $\alpha^i_j$ is $j^{th}$ step-evolution in $t_i$-direction and $\alpha_0^i=0$. Moreover, we define all possible permutations $\mathcal{P}$ of the functional measure such that
	\begin{align}
	    &\int_{\mathbf{q}(j\epsilon_1,\alpha^2_{j-1}\epsilon_2,\dots,\alpha^N_{j-1}\epsilon_N)}^{\mathbf{q}(j\epsilon_1,\alpha^2_j\epsilon_2,\dots,\alpha^N_j\epsilon_N)}\mathcal{P}\Bigg[\mathscr{D}[\mathbf{q}(t_2)],\mathscr{D}[\mathbf{q}(t_3)],\cdots,\mathscr{D}[\mathbf{q}(t_N)]\Bigg] \nonumber
	    \\&\;\;\;= \frac{1}{\mathcal{P}(N-1,r_j)}\big(\text{Summation of all possible permutations}\big)\;,
	\end{align}
	where $r_j=\sum_{i=1}^N \delta_{\alpha^i_{j-1},\alpha^i_j}$.
    \phantom{dd}
    \\
    \\
    What we have now for the multi-time propagator in terms of the parameterised variable $s$ is 
    \begin{equation}
         K(\mathbf{q}(\mathbf{t}(s'')),s''; \mathbf{q}(\mathbf{t}(s')), s') =\int_{\mathbf{q}(\mathbf{t}(s'))}^{\mathbf{q}(\mathbf{t}(s''))}\mathbb{D}[\mathbf{q}(\mathbf{t}(s));\Gamma\in\mathscr{B}]e^{\frac{i}{\hbar}\int_{\{\Gamma:\Gamma\in\mathscr{B}\}}\mathscr{L}}\;,\label{KK2}
    \end{equation}
    where $\int\mathbb{D}[\mathbf{q}(\mathbf{t}(s));\Gamma\in\mathscr{B}]$ measures the contribution of all possible paths $\mathbf{q}(\mathbf{t})$ and all possible paths $\Gamma\in\mathscr{B}$, where $\mathscr{B}$ is a family of paths connecting between $\mathbf{t}(s')$ and $\mathbf{t}(s'')$ on the space of time variables.
    \\
    \\
   Again, the propagator \eqref{KK2} is not appropriate for further consideration and, therefore, we shall apply the semi-classical approximation. Since we work with the parameterised time variable $s$, the action can be expanded in the same fashion with the single-time case, therefore the multi-time propagator now becomes
    \begin{eqnarray}
          K(\mathbf{q}'',s''; \mathbf{q}', s') 
          = e^{\frac{i}{\hbar}S[\mathbf{q}_c(s)]}\mathcal{Q}(\mathbf{q}'',s'', \mathbf{q}', s')\left[1+\mathcal{O}(\hbar)\right]\;,
    \end{eqnarray}
    where $\mathbf{q}_c$ is a classical solution and
    \begin{eqnarray}
          \mathcal{Q}(\mathbf{q}'',s'', \mathbf{q}', s') = \int_{\mathbf{y}(s')=0}^{\mathbf{y}(s'')=0}\mathbb{D}[\mathbf{y}(s);\Gamma]e^{\frac{i}{2\hbar}\int_{s'}^{s''}d\tau\int_{s'}^{s''}d\sigma\left(\mathbf{y}(\tau)\frac{\delta^2S[\mathbf{q}_c(s)]}{\delta\mathbf{q}(\tau)\delta\mathbf{q}(\sigma)}\mathbf{y}(\sigma)\right)}\;\label{factor for all path}
    \end{eqnarray}
    is a smooth function of end points. 
    \\
    \\
    Next, we will consider the multi-time propagator along an only arbitrary path $\Gamma$ connecting between end points $\mathbf{t}''$ and $\mathbf{t}'$ on space of time variables as follows
    \begin{eqnarray}
          K_\Gamma(\mathbf{q}(\mathbf{t}''),\mathbf{t}''; \mathbf{q}(\mathbf{t}'), \mathbf{t}') &=& e^{\frac{i}{\hbar}S_\Gamma[\mathbf{q}_c(\mathbf{t})]}\mathcal{Q}_\Gamma(\mathbf{q}'', \mathbf{q}', \mathbf{t}'', \mathbf{t}')\left[1+\mathcal{O}(\hbar)\right] \;,\label{propagator for arbitrary path}
    \end{eqnarray}
    where
    \begin{equation}
        \mathcal{Q}_\Gamma(\mathbf{q}'', \mathbf{q}', \mathbf{t}'', \mathbf{t}') = \int_{\mathbf{y}(\mathbf{t}')=0}^{\mathbf{y}(\mathbf{t}'')=0}\mathscr{D}_\Gamma[\mathbf{y}(\mathbf{t})]e^{\frac{i}{2\hbar}\sum_{j=1}^N\int_\Gamma du_j\int_\Gamma dv_j\left(\mathbf{y}(\mathbf{u})\frac{\delta^2S_{j,\Gamma}[\mathbf{q}_c(\mathbf{t})]}{\delta\mathbf{q}(\mathbf{u})\delta\mathbf{q}(\mathbf{v})}\mathbf{y}(\mathbf{v})\right)}\;.\label{factor for arbitrary path}
    \end{equation}
    Here $\int\mathscr{D}_\Gamma[\mathbf{y}]$ is a functional measure all possible fluctuations $\mathbf{y}$ along path $\Gamma$ on $N$-dimensional time-space and $S_{j,\Gamma}[\mathbf{q}(\mathbf{t})] = \int_\Gamma L_jdt_j$. Therefore, the function $\mathcal{Q}_\Gamma$ in \eqref{factor for arbitrary path} can be expressed, see appendix \ref{AppendixD}, in the form
    \begin{equation}
        \mathcal{Q}_\Gamma(\mathbf{q}'', \mathbf{q}', \mathbf{t}'', \mathbf{t}') = \det\left(\frac{i}{2\pi\hbar}\frac{\partial^2S_\Gamma [\mathbf{q}_c(\mathbf{t})]}{\partial\mathbf{q}(\mathbf{t}'')\partial\mathbf{q}(\mathbf{t}')}\right)^{\frac{1}{2}}\;,\label{formal factor for arbitrary path}
    \end{equation}
    where $S_\Gamma[\mathbf{q}(\mathbf{t})] = \int_\Gamma\mathscr{L}$.
    \\
    \\
    Here is an interesting point. We knew that the Lagrangian is not unique since different Lagrangian would produce an identical equation of motion. Moreover, in the context of integrable systems, a different set of Lagragians, producing the same equations of motion, may not all satisfy the closure relation, but only a special set of Lagrangians does. This structure would provide the space of possible Lagrangians. Consequently, in multi-time quantum systems, the propagator possessed the path independent feature on the space of time variables is the one that comes with a special set of Lagrangians satisfying the closure relation. This structure gives us an on top feature of the classical variational principle in the sense that this special set of Lagrangians plays a role of critical point resulting path independent propagator on the space of independent variables coined as the quantum variation \cite{SD.Kings}.
    \begin{theorem}\label{integrable for propagator}
	Let $\{L_1,L_2,...,L_N\}$ be a set of Lagrangians satisfying the Lagrangian closure relation and $\mathscr{L}=\sum_{j=1}^NL_jdt_j$ be the Lagrangian 1-form, where $L_j=L_j\left(\mathbf{q},\left\{\frac{\partial\mathbf{q}}{\partial t_j};j = 1, 2,\dots,N\right\};\mathbf{t}\right)$. On the space of independent variables (time variables), the multi-time propagator for any $\Gamma\in \mathscr{B}$, where $\mathscr{B}$ is a family of paths connecting between $\mathbf{t}'$ and $\mathbf{t}''$, gives equally contribution leading to 
	\begin{equation}
	    \oint\mathscr{D}_{C=\partial \mathcal{S}}[\mathbf{q}(\mathbf{t})]e^{\frac{i}{\hbar}\oint_{C=\partial \mathcal{S}}\mathscr{L}} = \mathbb{I}\;,
	\end{equation}
	where $\mathcal{S}$ is an arbitrary surface bounded by a contractible loop $C$ on the space of time variables, and therefore the multi-time quantum system is integrable.
	\end{theorem}
	\begin{proof}
	According to the equations \eqref{propagator for arbitrary path}-\eqref{formal factor for arbitrary path}, the propagator along path $\Gamma$ on space of time variables in the semi-classical limit reads
	\begin{eqnarray}
        K_\Gamma(\mathbf{q}(\mathbf{t}''),\mathbf{t}''; \mathbf{q}(\mathbf{t}'), \mathbf{t}') &=& \mathcal{Q}_\Gamma e^{\frac{i}{\hbar}S_\Gamma[\mathbf{q}_c(\mathbf{t})]}\left[1+\mathcal{O}(\hbar)\right]\;,\label{IP1}\\
        \mathcal{Q}_\Gamma&=&\det\left(\frac{i}{2\pi\hbar}\frac{\partial^2S_\Gamma [\mathbf{q}_c(\mathbf{t})]}{\partial\mathbf{q}(\mathbf{t}'')\partial\mathbf{q}(\mathbf{t}')}\right)^{\frac{1}{2}}\;,
	\end{eqnarray}
	and, for the path $\Gamma'$ with the same end points, the propagator reads
    \begin{eqnarray}
        K_{\Gamma'}(\mathbf{q}(\mathbf{t}''),\mathbf{t}''; \mathbf{q}(\mathbf{t}'), \mathbf{t}') &=& \mathcal{Q}_{\Gamma'}e^{\frac{i}{\hbar}S_{\Gamma'}[\mathbf{q}_c(\mathbf{t})]}\left[1+\mathcal{O}(\hbar)\right]\;,\label{IP2}\\
        \mathcal{Q}_{\Gamma'}&=&\det\left(\frac{i}{2\pi\hbar}\frac{\partial^2S_{\Gamma'} [\mathbf{q}_c(\mathbf{t})]}{\partial\mathbf{q}(\mathbf{t}'')\partial\mathbf{q}(\mathbf{t}')}\right)^{\frac{1}{2}}\;.
	\end{eqnarray}
	The closure relation for the classical Lagrangian 1-form $\mathscr{L}_c$ provides
	\begin{equation}
	    S_\Gamma [\mathbf{q}_c(\mathbf{t})]-S_{\Gamma'} [\mathbf{q}_c(\mathbf{t})]=\left(\int_\Gamma-\int_{\Gamma'}\right)\mathscr{L}_c = \oint_{C=\partial \mathcal{S}}\mathscr{L}_c=\iint_{\mathcal{S}}\sum_{k\geq 1}^N\sum_{l=1}^N\left(\frac{\partial L_l}{\partial t_k}-\frac{\partial L_k}{\partial t_l}\right)dt_k\wedge dt_l = 0\;,\label{II11}
	\end{equation}
	which is nothing but the path independent feature on independent variables space, see figure \ref{deform and loop}. Here $\mathcal{S}$ is an arbitrary surface bounded by a contractible loop $C$ on the space of time variables. Therefore, $\mathcal{Q}_\Gamma=\mathcal{Q}_{\Gamma'}$ and consequently we have
	\begin{equation}
	    K_\Gamma(\mathbf{q}(\mathbf{t}''),\mathbf{t}''; \mathbf{q}(\mathbf{t}'), \mathbf{t}') = K_{\Gamma'}(\mathbf{q}(\mathbf{t}''),\mathbf{t}''; \mathbf{q}(\mathbf{t}'), \mathbf{t}')\;, \label{eqllyK}
	\end{equation}
	where the $\mathcal{O}(\hbar)$ is ignored since there is extremely tiny contribution to the propagator in the semi-classical limit. 
    \\
    \\
    For a contractible loop $C=\partial \mathcal{S}$ on space of time variable, the propagator can be captured as
	\begin{equation}
	    K_{C=\partial \mathcal{S}} = \lim_{(\mathbf{t}'-\tilde{\mathbf{t}})\to 0}\int_{-\infty}^{\infty} d^Nq''\int_{-\infty}^{\infty} d^Nq'\det\left(\left(\frac{i}{2\pi\hbar}\right)^2\frac{\partial^2S_\Gamma [\mathbf{q}_c(\mathbf{t})]}{\partial\mathbf{q}(\mathbf{t}'')\partial\mathbf{q}(\mathbf{t}')}\frac{\partial^2(-S_{\Gamma'} [\mathbf{q}_c(\mathbf{t})])}{\partial\mathbf{q}(\mathbf{t}'')\partial\mathbf{q}(\tilde{\mathbf{t}})}\right)^{\frac{1}{2}}e^{\frac{i}{\hbar}\left(\int_{\mathbf{t}',\Gamma}^{\mathbf{t''}}-\int_{\tilde{\mathbf{t}},\Gamma'}^{\mathbf{t}''}\right)\mathscr{L}_c}\;.\label{loopPPTH}
	\end{equation}
	Then we write
	\begin{align}
	    S_c[\mathbf{q}',\mathbf{q}''] &=: \int_{\mathbf{t}'}^{\mathbf{t''}}\mathscr{L}_c\;,\\
	    S_c[\tilde{\mathbf{q}},\mathbf{q}''] &=: \int_{\tilde{\mathbf{t}}}^{\mathbf{t''}}\mathscr{L}_c\;.
	\end{align}
	Dropping out the subscripts $\Gamma$ and $\Gamma'$ because of path independent feature, we obtain\footnote{Here, $S_c[\mathbf{q}',\mathbf{q}'']$ is no longer a functional since the classical path $\mathbf{q}_c$ has been substituted. Thus, $S_c[\mathbf{q}',\mathbf{q}'']$ will be simply treated as a function depending on the initial positions.}
	\begin{equation}
	    \lim_{(\mathbf{t}'-\tilde{\mathbf{t}})\to 0}\left(\int_{\mathbf{t}'}^{\mathbf{t''}}-\int_{\tilde{\mathbf{t}}}^{\mathbf{t}''}\right)\mathscr{L}_c =\lim_{(\mathbf{t}'-\tilde{\mathbf{t}})\to 0} \sum_{i=1}^N\frac{S_c[\mathbf{q}',\mathbf{q}'']- S_c[\tilde{\mathbf{q}},\mathbf{q}'']}{\tilde q_i-q_i'}(\tilde q_i-q_i') = -\frac{\partial S_c}{\partial\tilde{\mathbf{q}}}\cdot(\tilde{\mathbf{q}}-\mathbf{q}')\;.\label{ID11}
	\end{equation}
	Substituting \eqref{ID11} into \eqref{loopPPTH}, the propagator \eqref{loopPPTH} becomes
	\begin{equation}
	    K_{C=\partial \mathcal{S}} = \left(\frac{1}{2\pi}\right)^N\int_{-\infty}^{\infty} d^Nq'\int_{-\infty}^\infty d^N\left(\frac{1}{\hbar}\frac{\partial S_c}{\partial\tilde{q}}\right)e^{-\frac{i}{\hbar}\frac{\partial S_c}{\partial\tilde{\mathbf{q}}}\cdot(\tilde{\mathbf{q}}-\mathbf{q}')}=\int_{-\infty}^{\infty} d^Nq'\delta^N(\tilde{q}-q') =\mathbb{I}\;.
	\end{equation}
	\end{proof}
	\noindent We shall point out a final feature of the multi-time propagator. From equation \eqref{unitarytime1}, it is not difficult to see that we could have a set of equations
	\begin{equation}
	    i\hbar \frac{\partial}{\partial t_j}K(\mathbf{q}(\mathbf t''),\mathbf{t}'';\mathbf{q}(\mathbf t'),\mathbf{t}')=\hat{\mathbf{H}}_jK(\mathbf{q}(\mathbf t''),\mathbf{t}'';\mathbf{q}(\mathbf t'),\mathbf{t}')\;,\;\text{where}\;j=1,2,...,N\;\text{and}\; \mathbf{t}''>\mathbf{t}'\; \;.\label{KKK}
	\end{equation}
	Again, the quantity $\partial_{t_j}-(1/i\hbar)\hat{\mathbf{H}}_j$, where $j=1,2,...,N$, can be treated as the covariant derivative and the system of equations \eqref{KKK} is overdetermined. Thus, a common nontrivial solution $K(\mathbf{q}(\mathbf t''),\mathbf{t}'';\mathbf{q}(\mathbf t'),\mathbf{t}')$ exists simultaneously if
	\begin{equation}
	    \frac{\partial}{\partial t_k}\frac{\partial}{\partial t_j}K(\mathbf{q}(\mathbf t''),\mathbf{t}'';\mathbf{q}(\mathbf t'),\mathbf{t}')=\frac{\partial}{\partial t_j}\frac{\partial}{\partial t_k}K(\mathbf{q}(\mathbf t''),\mathbf{t}'';\mathbf{q}(\mathbf t'),\mathbf{t}')\label{COMKK}
	\end{equation}
	holds. This compatibility \eqref{COMKK} gives again directly to the zero-curvature condition of the Hamiltonian operators \eqref{Sch0curve}.
	\section{Concluding discussion}\label{section5}
 In Schr\"{o}dinger picture, one can promote the set of Hamiltonians in the classical integrable systems to be a set of Hamiltonian operators and the set of Schr\"{o}dinger equations are obtained. This set of Schr\"{o}dinger equations is overdetermined and therefore a common non-trivial solution, wave function, exists if the Hamiltonian operators must satisfy the zero-curvature condition. The multi-time unitary operator can be expressed in terms of the Wilson line and possesses the path-independent feature. This means that, for the loop evolution, the unitary map in terms of the Wilson loop is identity. At this point, we may state that, for integrable quantum systems, the Hamiltonian operators must follow the zero-curvature condition, but the inverse is not necessary true, see \cite{Stefan}. In Feynman picture, the continuous multi-time propagator is derived. The interesting point is that this multi-time propagator comes with \emph{a new feature on sum over all possible paths}. One needs to take into account not only all possible paths on the space of dependent variables, but also on the space of independent variables(time variables). Of course, this idea is not new and it was first introduced by Nijhoff \cite{Franktalk} in 2013\footnote{This new perspective of treating the dependent and independent variables on the same equal footing was suggested in many places \cite{Atkinson,Rovelli}, see further discussion in \cite{Kingthesis}}. We point at this stage that what we come up for the formula of the continuous multi-time propagator in the 1-form case is not the same with Nijhoff's proposal. However, they do share the exactly the same interpretation. Another point is that, as we mention earlier on taking all possible path both dependent and independent variables, the propagator contains also non-classical paths which do not satisfy the closure relation. Then this new beauty beast must be tamed. Therefore, the semi-classical approximation is applied to the continuous multi-time propagator. The propapgator is then written in terms of the classical action together with the fluctuation (prefactor). With this new form of the continuous multi-time propagator (approximated one), the integrability criterion can be constructed. A major intriguing feature in this context is that there exists a space of Lagrangians. All Lagrangians produce the same equations of motion, but only a special set of Lagrangians satisfies the closure relation. With this special set of Lagrangians, the continuous multi-time propagator is extremum yielding path-independent feature on the space of independent variables. This new feature is known as the quantum variation \cite{SD.Kings}. The last point that we would like to mention is that our set up on deriving the continuous multi-time propagator is not only restricted to the quadratic Lagrangian cases. However, the result from King and Nijhoff \cite{SD.Kings} for the quadratic Lagrangians, namely harmonic oscillators, provides a solid verification of our formulation as the special case, see appendix \ref{AppendixE}. Therefore, further concrete examples are needed for non-quadratic Lagrangians.
	
	\newpage
	\appendix
	\section{The exponential maps for $N$-parameter group} \label{AppendixA}
	In this appendix, the full derivation of \eqref{U12} will be presented. For simplicity, we will first consider the 2 time variables: $\mathbf{t}=(t,\tau)$. The composite map $\hat{\mathbf{U}}_2\circ \hat{\mathbf{U}}_1$ could be written as:
	\begin{align}
	    \hat{\mathbf{U}}_2\circ\hat{\mathbf{U}}_1  &= \mathrm{T}e^{\int d\tau \hat{\mathbf{H}}_2(\mathbf{t})}\mathrm{T}e^{\int dt \hat{\mathbf{H}}_1(\mathbf{t})} \nonumber
	    \\&= \Bigg\{\mathrm{T}\Bigg[\mathbb{I}+\sum_{m=1}^\infty\frac{1}{m!}\Big(\prod_{j=1}^n\int d\tau_j\Big)\Big(\prod_{i=1}^n \hat{\mathbf{H}}_2(\mathbf{t}_i)\Big)\Bigg]\Bigg\}\Bigg\{\mathrm{T}\Bigg[\mathbb{I}+\sum_{n=1}^\infty\frac{1}{n!}\Big(\prod_{j=1}^n\int dt_j\Big)\Big(\prod_{i=1}^n \hat{\mathbf{H}}_1(\mathbf{t}_i)\Big)\Bigg]\Bigg\}\;.\label{expand composition}
	\end{align}
	The first two terms in \eqref{expand composition} would give 
	\begin{equation}
	    \int d\tau_1 \hat{\mathbf{H}}_2(\mathbf{t}_1)+\int dt_1 \hat{\mathbf{H}}_1(\mathbf{t}_1)=\int_\Gamma d\mathbf{t}_1\cdot \hat{\mathbf{H}}(\mathbf{t}_1)\;,
	\end{equation}
	where $\hat{\mathbf{H}}(\mathbf{t}_1) = (\hat{\mathbf{H}}_1(\mathbf{t}_1),\hat{\mathbf{H}}_2(\mathbf{t}_1))$. The next terms in the expansion will be
	\begin{eqnarray}
	    &&\frac{1}{2!}\Bigg(\int d\tau_1\int d\tau_2\mathrm{T}\big[\hat{\mathbf{H}}_2(\mathbf{t}_1)\hat{\mathbf{H}}_2(\mathbf{t}_2)\big]+\int d\tau_2\int dt_1\big[\hat{\mathbf{H}}_2(\mathbf{t}_2)\hat{\mathbf{H}}_1(\mathbf{t}_1)\big] \nonumber\\&&+\int d\tau_1\int dt_2\big[\hat{\mathbf{H}}_2(\mathbf{t}_1)\hat{\mathbf{H}}_1(\mathbf{t}_2)\big] 
	    +\int dt_1\int dt_2\mathrm{T}\big[\hat{\mathbf{H}}_1(\mathbf{t}_1)\hat{\mathbf{H}}_1(\mathbf{t}_2)\big]\Bigg)\;.\label{the 2nd composition}
	\end{eqnarray}
	Since the time variables of the second and third terms of the equation \eqref{the 2nd composition} have been ordered, we can insert the time-ordering operator into both of them. Using the fact that $[\hat{\mathbf{H}}_1,\hat{\mathbf{H}}_2]=0$, the equation \eqref{the 2nd composition} would become
	\begin{align}
	    &\frac{1}{2!}\mathrm{T}\Bigg(\int d\tau_1\int d\tau_2\big[\hat{\mathbf{H}}_2(\mathbf{t}_1)\hat{\mathbf{H}}_2(\mathbf{t}_2)\big]+\int d\tau_2\int dt_1\big[\hat{\mathbf{H}}_2(\mathbf{t}_2)\hat{\mathbf{H}}_1(\mathbf{t}_1)\big]\nonumber
	    \\&+\int d\tau_1\int dt_2\big[\hat{\mathbf{H}}_2(\mathbf{t}_1)\hat{\mathbf{H}}_1(\mathbf{t}_2)\big]
	    +\int dt_1\int dt_2\big[\hat{\mathbf{H}}_1(\mathbf{t}_1)\hat{\mathbf{H}}_1(\mathbf{t}_2)\big]\Bigg) \nonumber
	    \\&=\frac{1}{2!}\mathrm{T}\Bigg(\int d\tau_1\int d\tau_2\big[\hat{\mathbf{H}}_2(\mathbf{t}_1)\hat{\mathbf{H}}_2(\mathbf{t}_2)\big]+\int dt_1\int d\tau_2\big[\hat{\mathbf{H}}_1(\mathbf{t}_1)\hat{\mathbf{H}}_2(\mathbf{t}_2)\big] \nonumber
	    \\&\;+\int d\tau_1\int dt_2\big[\hat{\mathbf{H}}_2(\mathbf{t}_1)\hat{\mathbf{H}}_1(\mathbf{t}_2)\big]
	    +\int dt_1\int dt_2\big[\hat{\mathbf{H}}_1(\mathbf{t}_1)\hat{\mathbf{H}}_1(\mathbf{t}_2)\big]\Bigg) \nonumber
	    \\&=\frac{1}{2!}\mathrm{T}\Bigg\{\Big(\int d\tau_1 \hat{\mathbf{H}}_2(\mathbf{t}_1)+\int dt_1 \hat{\mathbf{H}}_1(\mathbf{t}_1)\Big)\Big(\int d\tau_2 \hat{\mathbf{H}}_2(\mathbf{t}_2) \nonumber
	    +\int dt_2 \hat{\mathbf{H}}_1(\mathbf{t}_2)\Big)\Bigg\}
	    \\&=\frac{1}{2!}\mathrm{T}\Bigg\{\int_\Gamma d\mathbf{t}_1\cdot \hat{\mathbf{H}}(\mathbf{t}_1)\int_\Gamma d\mathbf{t}_2\cdot \hat{\mathbf{H}}(\mathbf{t}_2)\Bigg\}\;.
	\end{align}
	For further terms in the expansion of \eqref{expand composition}, we apply the same trick. Finally, we obtain
	\begin{eqnarray}
	    \hat{\mathbf{U}}_2\circ\hat{\mathbf{U}}_1  &=& \mathrm{T}\Bigg[\mathbb{I}+\sum_{m=1}^\infty\frac{1}{m!}\Big(\prod_{j=1}^n\int_\Gamma d\mathbf{t}_j\cdot \hat{\mathbf{H}}(\mathbf{t}_j)\Big)\Bigg] \nonumber
	    \\&=& \mathrm{T}e^{\int_\Gamma d\mathbf{t}\cdot \hat{\mathbf{H}}(\mathbf{t})}\;.
	\end{eqnarray}
	\section{The definition of the time-ordering operator for a loop evolution}\label{Appendix0}
    Basically, in the case of a single time, a time-ordering operator is given by
    \begin{equation}
        \mathrm{T}\left[\hat{\mathbf{H}}(t_1)\hat{\mathbf{H}}(t_2)\right] = \Theta(t_1-t_2)\hat{\mathbf{H}}(t_1)\hat{\mathbf{H}}(t_2)+\Theta(t_2-t_1)\hat{\mathbf{H}}(t_2)\hat{\mathbf{H}}(t_1)\;,
    \end{equation}
    where $\Theta(t_1-t_2)$ is a Heaviside step function. The mathematical object that we are interested in this situation is the unitary operator. For the forward evolution and the backward evolution given in figure \ref{Tsloop}, the unitary operators are given by
    \begin{align}
        \hat{\mathbf{U}}_{\Gamma} =& \mathrm{T}e^{-\frac{i}{\hbar}\int_{t'}^{t''}\hat{\mathbf{H}}(t)dt}=\mathrm{T}e^{-\frac{i}{\hbar}\int_{\Gamma}\hat{\mathbf{H}}(t)dt}\;,\label{UG}
        \\\hat{\mathbf{U}}_{\Gamma'} =& \mathrm{T}e^{-\frac{i}{\hbar}\int_{t''}^{t'}\hat{\mathbf{H}}(t)dt} =\mathrm{T}e^{-\frac{i}{\hbar}\int_{\Gamma'}\hat{\mathbf{H}}(t)dt}= \mathrm{T}e^{\frac{i}{\hbar}\int_{t'}^{t''}\hat{\mathbf{H}}(t)dt}\;,\label{UG'}
    \end{align}
    respectively.
    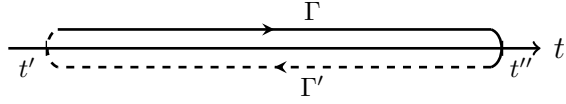
\begin{figure}[h]
        \centering
        \begin{tikzpicture}[]
        	        \path[line width = 1pt, draw = black,->]
 	                  (0,3) to (7,3);
 	                  
 	                 \path[line width = 1pt, draw = black, postaction ={on each segment = {mid arrow}}]
 	                  (0.65,3.25) to (6.35,3.25);
 	                 
 	                 \path[dashed, line width = 1pt, draw = black, postaction ={on each segment = {mid arrow}}]
 	                  (6.35,2.75) to (0.65,2.75);
 	                  
 	                   \path[line width = 1pt, draw = black]
 	                  (6.35,3.25) to[bend left = 70] (6.35,2.75); 
 	                  
 	                  \path[dashed, line width = 1pt, draw = black]
 	                  (0.65,2.75) to[bend left = 70] (0.65,3.25);
 	                  
 	                 \node at (7.25,3) {$t$};
 	                 
 	                 \path[line width = 0.5pt, draw = black]
 	                  (0.5,2.9) to (0.5,3.1);
 	                  
 	                  \path[line width = 0.5pt, draw = black]
 	                  (6.5,2.9) to (6.5,3.1);
 	                  
 	                 \node at (0.25,2.75) {\footnotesize{$t'$}};
 	                 \node at (6.75,2.75) {\footnotesize{$t''$}};
 	                 
 	                 \node at (4,3.5) {\footnotesize{$\Gamma$}};
 	                 \node at (4,2.5) {\footnotesize{$\Gamma'$}};
 	                
 	            \end{tikzpicture}
        \caption{The loop evolution on the space of a single time variable}
        \label{Tsloop}
    \end{figure}
    \\
    Since the unitary operator possesses a property $\hat{\mathbf{U}}_{\Gamma}\hat{\mathbf{U}}_{\Gamma'}=\mathbb{I}$, therefore, what we have now is
    \begin{align}
        \hat{\mathbf{U}}_{\Gamma}\hat{\mathbf{U}}_{\Gamma'} =& \left(\mathrm{T}e^{-\frac{i}{\hbar}\int_{t'}^{t''}\hat{\mathbf{H}}(t)dt}\right)\left(\mathrm{T}e^{\frac{i}{\hbar}\int_{t'}^{t''}\hat{\mathbf{H}}(t)dt}\right) \nonumber
        \\=&\left(\mathbb{I}-\frac{i}{\hbar}\int_{t'}^{t''}\hat{\mathbf{H}}(t)dt+\frac{1}{2!}\left(-\frac{i}{\hbar}\right)^2\int_{t'}^{t''}dt_1\int_{t'}^{t''}dt_2\mathrm{T}\left[\hat{\mathbf{H}}(t_1)\hat{\mathbf{H}}(t_2)\right]+...\right)\nonumber\\
        &\times\left(\mathbb{I}+\frac{i}{\hbar}\int_{t'}^{t''}\hat{\mathbf{H}}(t)dt+\frac{1}{2!}\left(\frac{i}{\hbar}\right)^2\int_{t'}^{t''}dt_1\int_{t'}^{t''}dt_2\mathrm{T}\left[\hat{\mathbf{H}}(t_1)\hat{\mathbf{H}}(t_2)\right]+...\right)\nonumber 
        \\=&\;\;\mathbb{I}-\frac{i}{\hbar}\left(\int_{t'}^{t''}dt-\int_{t'}^{t''}dt\right)\hat{\mathbf{H}}(t) \nonumber
        \\&+\frac{1}{2!}\left(-\frac{i}{\hbar}\right)^2\left(\int_{t'}^{t''}dt_1-\int_{t'}^{t''}dt_1\right)\left(\int_{t'}^{t''}dt_2-\int_{t'}^{t''}dt_2\right)\mathrm{T}\left[\hat{\mathbf{H}}(t_1)\hat{\mathbf{H}}(t_2)\right]+...\nonumber
        \\=&\;\; \mathrm{T}\left(e^{-\frac{i}{\hbar}\int_{t'}^{t''}\hat{\mathbf{H}}(t)dt+\frac{i}{\hbar}\int_{t'}^{t''}\hat{\mathbf{H}}(t)dt}\right) = \mathbb{I}\;.\label{NT}
    \end{align}
    This suggests that the time-ordering operators in \eqref{UG} and \eqref{UG'} have the same structure.
    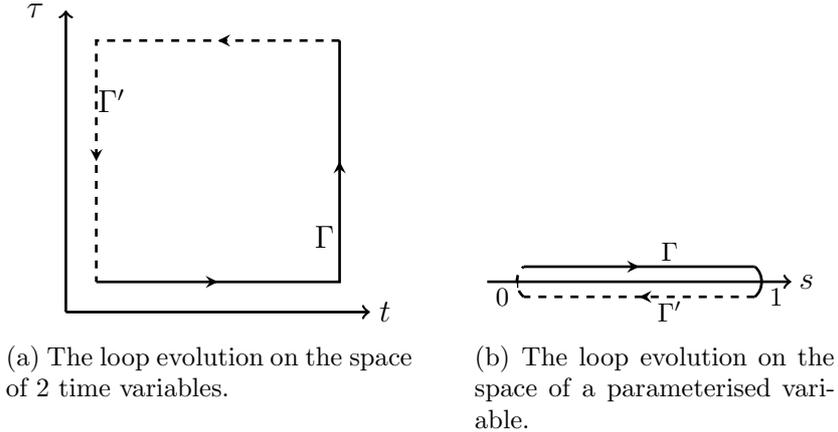
\begin{figure}[h]
    	    \centering
    	    \subfloat[The loop evolution on the space of 2 time variables.]{\label{Tloop1}
        	    \begin{tikzpicture}[scale = 0.8]
            	     \path[line width = 1pt, draw = black,->]
 	                  (0,0) to (5,0);
 	                  
 	                  \path[line width = 1pt, draw = black,->]
 	                  (0,0) to (0,5);
 	                  
 	                 \path[line width = 1pt, draw = black, postaction ={on each segment = {mid arrow}}]
 	                  (0.5,0.5) to (4.5,0.5) to (4.5,4.5);
 	                 
 	                 \path[dashed, line width = 1pt, draw = black, postaction ={on each segment = {mid arrow}}]
 	                  (4.5,4.5) to (0.5,4.5) to (0.5,0.5);
 	                  
 	                 \node at (5.25,0) {$t$};
 	                 \node at (-0.5,5) {$\tau$};
 	                 
 	                 \node at (4.25,1.25) {$\Gamma$};
 	                 \node at (0.75,3.5) {$\Gamma'$};
 	                 
 	            \end{tikzpicture}
 	       }
 	       \qquad
 	        \subfloat[The loop evolution on the space of a parameterised variable.]{\label{Tloop2}
        	    \begin{tikzpicture}[scale = 0.8]
        	        \path[line width = 1pt, draw = black,->]
 	                  (0,3) to (5,3);
 	                  
 	                 \path[line width = 1pt, draw = black, postaction ={on each segment = {mid arrow}}]
 	                  (0.6,3.25) to (4.4,3.25);
 	                 
 	                 \path[dashed, line width = 1pt, draw = black, postaction ={on each segment = {mid arrow}}]
 	                  (4.4,2.75) to (0.6,2.75);
 	                  
 	                  \path[line width = 1pt, draw = black]
 	                  (4.4,3.25) to[bend left = 50] (4.4,2.75); 
 	                  
 	                  \path[dashed, line width = 1pt, draw = black]
 	                  (0.6,2.75) to[bend left = 50] (0.6,3.25);
 	                  
 	                 \node at (5.25,3) {$s$};
 	                 
 	                 \path[line width = 0.5pt, draw = black]
 	                  (0.5,2.9) to (0.5,3.1);
 	                  
 	                  \path[line width = 0.5pt, draw = black]
 	                  (4.5,2.9) to (4.5,3.1);
 	                  
 	                 \node at (0.25,2.75) {\footnotesize{$0$}};
 	                 \node at (4.75,2.75) {\footnotesize{$1$}};
 	                 
 	                 \node at (3,3.5) {\footnotesize{$\Gamma$}};
 	                 \node at (3,2.5) {\footnotesize{$\Gamma'$}};
 	                
 	            \end{tikzpicture}
 	         }
 	         \caption{The loop evolution of system for 2-time structure.}
 	         \label{Tloop}
    \end{figure}
    \\
    Next, we extend the idea to the case of the multi-time situation. For simplicity, we shall consider the case of the $2$-dimensional time space: $\mathbf{t}=(t,\tau)$. For the paths $\Gamma$ and $\Gamma'$ in figure \ref{Tloop1}, the unitary multi-time evolution operators can be presented as follows:
    \begin{align}
        \hat{\mathbf{U}}_{\Gamma} =& \mathrm{T}e^{-\frac{i}{\hbar}\int_{\Gamma}\hat{\mathbf{H}}_{t}(\mathbf{t})dt+\hat{\mathbf{H}}_{\tau}(\mathbf{t})d\tau}\;,\label{2UG}
        \\\hat{\mathbf{U}}_{\Gamma'} =& \mathrm{T}e^{-\frac{i}{\hbar}\int_{\Gamma'}\hat{\mathbf{H}}_{t}(\mathbf{t})dt+\hat{\mathbf{H}}_{\tau}(\mathbf{t})d\tau}\;.\label{2UG'}
    \end{align}
    Here, we introduce a new variable $s$ such that $\mathbf{t}(s)=(t(s),\tau(s))$, where $0 \le s \le 1$, see figure \ref{Tloop2}. Then the unitary multi-time evolution operators in \eqref{2UG} and \eqref{2UG'} become
     \begin{align}
        \hat{\mathbf{U}}_{\Gamma} =&\mathrm{T}e^{-\frac{i}{\hbar}\int_{0}^1\hat{\mathbf H}(s)ds} =\mathrm{T}e^{-\frac{i}{\hbar}\int_{\Gamma}\hat{\mathbf H}(s)ds}\;,\label{3UG}
        \\
        \hat{\mathbf{U}}_{\Gamma'} =& \mathrm{T}e^{-\frac{i}{\hbar}\int_{1}^0\hat{\mathbf H}(s)ds} =\mathrm{T}e^{-\frac{i}{\hbar}\int_{\Gamma'}\hat{\mathbf H}(s)ds}\;,\label{3UG'}
    \end{align}
    where $\hat{\mathbf H}(s)=\hat{\mathbf{H}}_{t}(\mathbf{t})dt+\hat{\mathbf{H}}_{\tau}(\mathbf{t})d\tau$. 
    With this a single parameterised variable $s$, one can apply the same process as in \eqref{NT} and this suggests that the nature of the time-ordering operator in the case of the multi-time and the single time is effectively identical to explain the loop evolution. Therefore, $\hat{\mathbf{U}}_{\Gamma}\hat{\mathbf{U}}_{\Gamma'}=\mathbb{I}$ gives us a commutativity of the multi-time evolution or the integrability condition.
    The idea can be directly extended into the case of $N$ time variables with this parameterised method.
    
    \section{Eliminating the redundant path of $K^{(3)}$}\label{AppendixB}
    In equation \eqref{3first propagator complete}, the order the summation is crucial to avoid the redundant path. To see this, we first consider a point $m_1=0$. The variable $n_1$ will not contribute because of collapsing of the integration as follows
    \begin{align}
         K^{(3)} =& \sum_{n_1=1}^{\mathcal{N}-1}\int_{(0,0)}^{(n_1,0)}\mathscr{D}[\mathbf{q}(t_1,0)]\int_{-\infty}^{\infty}d^N\mathbf{q}(n_1,0)\int_{(n_1,0)}^{(n_1,0)}\mathscr{D}[\mathbf{q}(n_1,t_2)]\int_{-\infty}^{\infty}d^N\mathbf{q}(n_1,0)\int_{(n_1,0)}^{(\mathcal{N},0)}\mathscr{D}[\mathbf{q}(t_1,0)]\nonumber
         \\&\times\int_{-\infty}^{\infty}d^N\mathbf{q}(\mathcal{N},0)\int_{(\mathcal{N},0)}^{(\mathcal{N},\mathcal{N})}\mathscr{D}[\mathbf{q}(\mathcal{N},t_2)] e^{\frac{i}{\hbar}\int_{s'}^{s''}\mathscr{L}} \nonumber
         \\=& \bra{\mathbf{q}(\mathcal{N},\mathcal{N})}\hat{\mathbf{U}}(\mathcal{N},\mathcal{N};\mathcal{N},0)\hat{\mathbf{U}}(\mathcal{N},0;n_1,0)\hat{\mathbf{U}}(n_1,0;n_1,0)\hat{\mathbf{U}}(n_1,0;0,0)\ket{\mathbf{q}(0,0)}\nonumber
         \\=& \bra{\mathbf{q}(\mathcal{N},\mathcal{N})}\hat{\mathbf{U}}(\mathcal{N},\mathcal{N};\mathcal{N},0)\hat{\mathbf{U}}(\mathcal{N},0;0,0)\ket{\mathbf{q}(0,0)}\nonumber
         \\=& \int_{(0,0)}^{(\mathcal{N},0)}\mathscr{D}[\mathbf{q}(t_1,0)]\int_{-\infty}^{\infty}d^N\mathbf{q}(\mathcal{N},0)\int_{(\mathcal{N},0)}^{(\mathcal{N},\mathcal{N})}\mathscr{D}[\mathbf{q}(\mathcal{N},t_2)]e^{\frac{i}{\hbar}\int_{s'}^{s''}\mathscr{L}}\;,\label{identity of propagator2}
    \end{align}
    which is the propagator in the equation \eqref{K^1}. Hence, the variable $n_1$ in the first line of equation \eqref{identity of propagator2} is arbitrary. On the other hand, if the summation over $n_1$ is considered first, there are redundant paths for every $n_1$ at $m_1=0$.
	\section{Eliminating the redundant path of $K^{(5)}$}\label{AppendixC}
	For $m_1=0$ and $m_2=m_1$ of $K^{(5)}$, we obtain the repeated paths such as the figure \ref{repeated paths}.
	\begin{figure}[h]
    	    \centering
    	    \subfloat[$m_1=0$ with $m_2\neq0$]{\label{m_1=0}
        	    \begin{tikzpicture}[scale = 0.4]
            	    \path[dashed, draw=black]
 	                (0,0) to (10,0)
 	                (10,0) to (10,10)
 	                (10,10) to (0,10)
 	                (0,10) to (0,0);
 	                
 	                \path[line width = 1 pt, draw=black, postaction ={on each segment = {mid arrow}}]
 	                (0,0) to (3,0)
 	                (3,0) to (6,0)
 	                (6,0) to (6,5)
 	                (6,5) to (10,5)
 	                (10,5) to (10,10);
 	                
 	                \node[circle, fill, inner sep=1.25 pt] at (0,0) {};
 	                \node[circle, fill, inner sep=1.25 pt] at (10,10) {};
 	                \node[circle, fill, inner sep=1.25 pt] at (3,0) {};
 	                \node[circle, fill, inner sep=1.25 pt] at (6,0) {};
 	                \node[circle, fill, inner sep=1.25 pt] at (6,5) {};
 	                \node[circle, fill, inner sep=1.25 pt] at (10,5) {};
 	                
 	                \node at (0,-0.5) {\scriptsize{$(0,0)$}};
 			        \node at (3,0.5) {\scriptsize{$(n_1,0)$}};
 			        \node at (6,-0.5) {\scriptsize{$(n_2,0)$}};
 			        \node at (6,5.5) {\scriptsize{$(n_2,m_2)$}};
 			        \node at (10,4.5) {\scriptsize{$(\mathcal{N},m_2)$}};
 			        \node at (10,10.5) {\scriptsize{$(\mathcal{N},\mathcal{N})$}};
 
 	            \end{tikzpicture}
 	       }
 	       \qquad
 	        \subfloat[$m_1=m_2\neq0$]{\label{m_1=m_2}
        	    \begin{tikzpicture}[scale = 0.4]
            	    \path[dashed, draw=black]
 	                (0,0) to (10,0)
 	                (10,0) to (10,10)
 	                (10,10) to (0,10)
 	                (0,10) to (0,0);
 	                
 	                \path[line width = 1 pt, draw=black, postaction ={on each segment = {mid arrow}}]
 	                (0,0) to (6,0)
 	                (6,0) to (6,5)
 	                (6,5) to (8,5)
 	                (8,5) to (10,5)
 	                (10,5) to (10,10);
 	                
 	                \node[circle, fill, inner sep=1.25 pt] at (0,0) {};
 	                \node[circle, fill, inner sep=1.25 pt] at (10,10) {};
 	                \node[circle, fill, inner sep=1.25 pt] at (8,5) {};
 	                \node[circle, fill, inner sep=1.25 pt] at (6,0) {};
 	                \node[circle, fill, inner sep=1.25 pt] at (6,5) {};
 	                \node[circle, fill, inner sep=1.25 pt] at (10,5) {};
 	                
 	                \node at (0,-0.5) {\scriptsize{$(0,0)$}};
 			        \node at (8,4.5) {\scriptsize{$(n_2,m_2)$}};
 			        \node at (6,-0.5) {\scriptsize{$(n_1,0)$}};
 			        \node at (6,5.5) {\scriptsize{$(n_1,m_1)$}};
 			        \node at (10,5.5) {\scriptsize{$(\mathcal{N},m_2)$}};
 			        \node at (10,10.5) {\scriptsize{$(\mathcal{N},\mathcal{N})$}};
 	            \end{tikzpicture}
 	         }
 	         \caption{The example of repeated paths between cases of $m_1=0$ with $m_2\in[1,\mathcal{N}-1]$ and $m_1=m_2\neq0$.}\label{repeated paths}
 	    \end{figure}
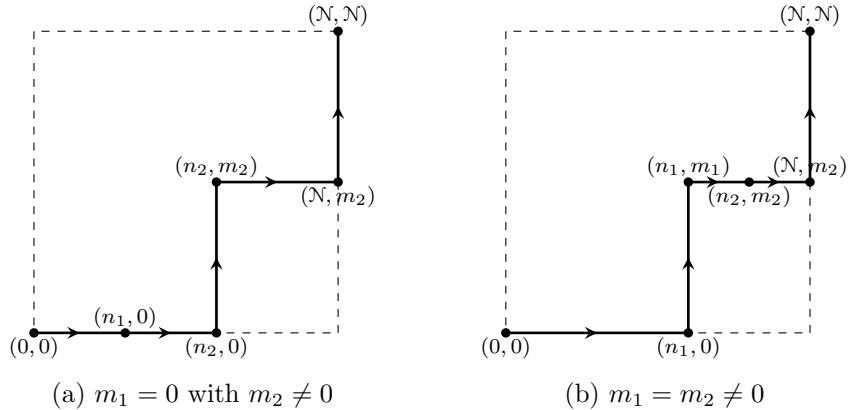
 	 \\
 	In figure \ref{m_1=0}, $n_1$ is in the range $[1,n_2]$ and therefore $n_2\ge 2$. We obtain
 	 \begin{align}
         K^{(5)}_{a} =& \sum_{n_2=2}^{\mathcal{N}-1}\sum_{m_2=1}^{\mathcal{N}}\int_{(0,0)}^{(n_2,0)}\mathscr{D}[\mathbf{q}(t_1,0)]\int_{-\infty}^{\infty}d^N\mathbf{q}(n_2,0)\int_{(n_2,0)}^{(n_2,m_2)}\mathscr{D}[\mathbf{q}(n_2,t_2)]\int_{-\infty}^{\infty}d^N\mathbf{q}(n_2,m_2)\nonumber
         \\&\times\int_{(n_2,m_2)}^{(\mathcal{N},m_2)}\mathscr{D}[\mathbf{q}(t_1,m_2)]\int_{-\infty}^{\infty}d^N\mathbf{q}(\mathcal{N},m_2)\int_{(\mathcal{N},m_2)}^{(\mathcal{N},\mathcal{N})}\mathscr{D}[\mathbf{q}(\mathcal{N},t_2)]e^{\frac{i}{\hbar}\int_{s'}^{s''}\mathscr{L}}\;.\label{5first propagator m_1=0}
    \end{align}
    In figure \ref{m_1=m_2}, the path between points $(n_2,m_1)$ and $(n_2,m_2)$ is collapsed as a point at $(n_2,m_2)$. And, if $n_1\geq2$, it will give exactly the same path with figure \ref{m_1=0}. Therefore, the propagator is
    \begin{align}
         K^{(5)}_{b} =& \sum_{n_1=2}^{\mathcal{N}-1}\sum_{m_1=1}^{\mathcal{N}}\int_{(0,0)}^{(n_1,0)}\mathscr{D}[\mathbf{q}(t_1,0)]\int_{-\infty}^{\infty}d^N\mathbf{q}(n_1,0)\int_{(n_1,0)}^{(n_1,m_1)}\mathscr{D}[\mathbf{q}(n_1,t_2)]\int_{-\infty}^{\infty}d^N\mathbf{q}(n_1,m_1)\nonumber
         \\&\times\int_{(n_1,m_1)}^{(\mathcal{N},m_1)}\mathscr{D}[\mathbf{q}(t_1,m_1)] \int_{-\infty}^{\infty}d^N\mathbf{q}(\mathcal{N},m_1)\int_{(\mathcal{N},m_1)}^{(\mathcal{N},\mathcal{N})}\mathscr{D}[\mathbf{q}(\mathcal{N},t_2)]e^{\frac{i}{\hbar}\int_{s'}^{s''}\mathscr{L}}\;.\label{5first propagator m_1=m_2}
    \end{align}
 	To avoid this problem, we have to set $n_1=1$ then the propagator \eqref{5first propagator m_1=m_2} would vanish since it out of the case.
 	\\
 	\\
 	For $K^{(7)}$, we would fix $n_1=1$, $n_2=2$ and $n_3=[3,\mathcal{N}-1]$ to avoid the redundant path. And, for $K^{(9)}$, we would fix $n_1=1$, $n_2=2$, $n_3=3$ and $n_4=[4,\mathcal{N}-1]$ to settle the problem. Finally, for $K^{(\text{All})}$, every single $n_i$ will be fixed, where $i=1,2,\dots,\mathcal{N}-1$ to eliminate the redundancy, then the sum over $n_i$ will disappear, resulting in the propagator \eqref{all t1 propagator}.
	\section{The explicit form of $\mathcal{Q}_\Gamma$ in the equation \eqref{factor for arbitrary path}}\label{AppendixD}
	Recalling the equation \eqref{factor for arbitrary path}
	\begin{equation}
        \mathcal{Q}_\Gamma(\mathbf{q}'', \mathbf{q}', \mathbf{t}'', \mathbf{t}') = \int_{\mathbf{y}(\mathbf{t}')=0}^{\mathbf{y}(\mathbf{t}'')=0}\mathscr{D}_\Gamma[\mathbf{y}(\mathbf{t})]e^{\frac{i}{2\hbar}\sum_{j=1}^N\int_\Gamma du_j\int_\Gamma dv_j\left(\mathbf{y}(\mathbf{u})\frac{\delta^2S_{j,\Gamma}[\mathbf{q}_c(\mathbf{t})]}{\delta\mathbf{q}(\mathbf{u})\delta\mathbf{q}(\mathbf{v})}\mathbf{y}(\mathbf{v})\right)}\nonumber\;,
    \end{equation}
    we then treat
    \begin{equation}
        \mathbf{y}(\mathbf{u})\frac{\delta^2S_{j,\Gamma}[\mathbf{q}_c(\mathbf{t})]}{\delta\mathbf{q}(\mathbf{u})\delta\mathbf{q}(\mathbf{v})}\mathbf{y}(\mathbf{v})
        =: \mathbf{y}(\mathbf{u})\Big(-\delta(\mathbf{u}-\mathbf{v})\hat{\mathbf{O}}_j\Big)\mathbf{y}(\mathbf{v})\;.
    \end{equation}
    We write $\mathbf{y}(\mathbf{t})=\sum_na_n\mathbf{y}_{n,\Gamma}(\mathbf{t})$, where $a_n$ now represent fluctuation and $\mathbf{y}_{n,\Gamma}$ are eigenbases associated with the path $\Gamma$ satisfying
    \begin{equation}
        \int_\Gamma dt_i \mathbf{y}_{n,\Gamma}(\mathbf{t})\mathbf{y}_{m,\Gamma}(\mathbf{t})\big|_{\text{fix}\;t_{j\neq i}} = \delta_{nm}\;,\label{orthonormal}
    \end{equation}
    and the boundary conditions
    \begin{equation}
        \mathbf{y}_{n,\Gamma}(\mathbf{t}'') = \mathbf{y}_{n,\Gamma}(\mathbf{t}') = 0\;.
    \end{equation}
    The equation \eqref{factor for arbitrary path} becomes
    \begin{equation}
        \mathcal{Q}_\Gamma(\mathbf{q}'', \mathbf{q}', \mathbf{t}'', \mathbf{t}') = \int\mathscr{D}[a_n]e^{-\frac{i}{2\hbar}\sum_{j=1}^N\big(\sum_n(\lambda_j)_{n,\Gamma}\abs{a_n}^2\big)}\;,\label{QQ1}
    \end{equation}
    where $(\lambda_j)_{n,\Gamma}$ are eigenvalues associated with the eigenbases $\mathbf{y}_{n,\Gamma}$ for the operator $\hat{\mathbf{O}}_j$. Inserting the equation \eqref{orthonormal} into \eqref{QQ1}, we obtain
    \begin{eqnarray}
        \mathcal{Q}_\Gamma(\mathbf{q}'', \mathbf{q}', \mathbf{t}'', \mathbf{t}') &=& \int\mathscr{D}[a_n]e^{-\frac{i}{2\hbar}\sum_{j=1}^N\big(\sum_n(\lambda_j)_{n,\Gamma}\abs{a_n}^2\big)} \nonumber
        \\&=&\int\mathscr{D}[a_n]e^{-\frac{i}{2\hbar}\sum_{j=1}^N\big(\sum_n\sum_m(\lambda_j)_{n,\Gamma}a_na_m\delta_{mn}\big)} \nonumber
        \\&=&\int\mathscr{D}[a_n]e^{-\frac{i}{2\hbar}\sum_{j=1}^N\big(\int_\Gamma dt_l\left(\sum_ma_m\mathbf{y}_{m,\Gamma}(\mathbf{t})\right)\left(\sum_n(\lambda_j)_{n,\Gamma}a_n\mathbf{y}_{n,\Gamma}(\mathbf{t})\right)\big|_{\text{fix}\;t_{j\neq i}}\big)} \nonumber
        \\&=&\int\mathscr{D}[a_n]e^{-\frac{i}{2\hbar}\sum_{j=1}^N\big(\int_\Gamma dt_l\left(\sum_ma_m\mathbf{y}_{m,\Gamma}(\mathbf{t})\right)\hat{\mathbf{O}}_j\left(\sum_na_n\mathbf{y}_{n,\Gamma}(\mathbf{t})\right)\big|_{\text{fix}\;t_{j\neq i}}\big)} \nonumber
        \\&=& \int\mathscr{D}_\Gamma[\mathbf{y}(\mathbf{t})]e^{-\frac{i}{2\hbar}\int_\Gamma dt_l\mathbf{y}(\mathbf{t})\big(\sum_{j=1}^N\hat{\mathbf{O}}_j\big)\mathbf{y}(\mathbf{t})} \nonumber
        \\&=& \int\mathscr{D}_\Gamma[\mathbf{y}(\mathbf{t})]e^{\frac{i}{2\hbar}\int_\Gamma du_l \int_\Gamma dv_l \left(\mathbf{y}(\mathbf{u})\frac{\delta^2}{\delta\mathbf{q}(\mathbf{u})\delta\mathbf{q}(\mathbf{v})}\left(\sum_{j=1}^NS_{j,\Gamma}[\mathbf{q}_c(\mathbf{t})]\right)\mathbf{y}(\mathbf{v})\right)}\nonumber
        \\&=& \int\mathscr{D}_\Gamma[\mathbf{y}(\mathbf{t})]e^{\frac{i}{2\hbar}\int_\Gamma du_l \int_\Gamma dv_l \left(\mathbf{y}(\mathbf{u})\frac{\delta^2S_{\Gamma}[\mathbf{q}_c(\mathbf{t})]}{\delta\mathbf{q}(\mathbf{u})\delta\mathbf{q}(\mathbf{v})}\mathbf{y}(\mathbf{v})\right)}
    \end{eqnarray}
    which is identical with the single-time case, see \eqref{factor for single path}. Therefore, we obtain
    \begin{equation}
        \mathcal{Q}_\Gamma(\mathbf{q}'', \mathbf{q}', \mathbf{t}'', \mathbf{t}') = \det\left(\frac{i}{2\pi\hbar}\frac{\partial^2S_\Gamma [\mathbf{q}_c(\mathbf{t})]}{\partial\mathbf{q}(\mathbf{t}'')\partial\mathbf{q}(\mathbf{t}')}\right)^{\frac{1}{2}}\;.
    \end{equation}
    We note that the existence of bases $\mathbf{y}_{n,\Gamma}(\mathbf{t})$, satisfying the equation \eqref{orthonormal}, is nontrivial. However, an explicit example is illustrated in appendix \ref{AppendixE}.
    \section{The path-independent feature of the propagator for the two harmonic oscillators}\label{AppendixE}
    Here we will show that, with a special set of quadratic Lagrangians satisfying the closure relation, the multi-time propagator possesses the path independent feature on the space of independent variables. We first write the set of Lagrangians $\{L_1, L_2\}$, giving the equations of motion in \eqref{quad1} and \eqref{quad2} and  satisfying the Lagrangian closure relation, as
    \begin{eqnarray}
        L_1=\frac{1}{2}\left(\frac{\partial \mathbf{q}}{\partial t_1}\right)^2-\frac{\omega_1^2 \mathbf{q}^2}{2}\;,
        \\L_2=\frac{1}{2}\left(\frac{\partial \mathbf{q}}{\partial t_2}\right)^2-\frac{\omega_2^2 \mathbf{q}^2}{2}\;,
    \end{eqnarray}
    where $\omega_{1,2}$ are constants and $ \mathbf{q}(t_1,t_2) = (q_1(t_1,t_2)\;q_2(t_1,t_2))$.
    \\
    \\
    Then, in order to verify path independent feature, we consider the exponent term in \eqref{factor for arbitrary path}
    \begin{eqnarray}
        &&\frac{i}{2\hbar}\sum_{j=1}^2\int_\Gamma du_j\int_\Gamma dv_j\left(\mathbf{y}(\mathbf{u})\frac{\delta^2S_{j,\Gamma}[ \mathbf{q}_c(\mathbf{t})]}{\delta \mathbf{q}(\mathbf{u})\delta  \mathbf{q}(\mathbf{v})}\mathbf{y}(\mathbf{v})\right) \nonumber 
        \\&&= \frac{i}{2\hbar}\left(\int_\Gamma dt_1\left[\left(\frac{\partial \mathbf{y}}{\partial t_1}\right)^2-\omega_1^2\mathbf{y}^2\right]+dt_2\left[\left(\frac{\partial \mathbf{y}}{\partial t_2}\right)^2-\omega_2^2\mathbf{y}^2\right]\right)\;.
    \end{eqnarray}
    \begin{figure}[h]
    	    \centering
    	    \subfloat[path A]{\label{path A}
        	    \begin{tikzpicture}[scale = 0.4]
            	    \path[dashed, draw=black]
 	                (0,0) to (10,0)
 	                (10,0) to (10,10)
 	                (10,10) to (0,10)
 	                (0,10) to (0,0);
 	                
 	                \path[line width = 1 pt, draw=black, postaction ={on each segment = {mid arrow}}]
 	                (0,0) to (10,0)
 	                (10,0) to (10,10);
 	                
 	                \node[circle, fill, inner sep=1.25 pt] at (0,0) {};
 	                \node[circle, fill, inner sep=1.25 pt] at (10,10) {};
 	                \node[circle, fill, inner sep=1.25 pt] at (10,0) {};
 	                
 	                \node at (0,-0.5) {\footnotesize{$(0,0)$}};
 			        \node at (10,-0.5) {\footnotesize{$(T_1,0)$}};
 			        \node at (10,10.5) {\footnotesize{$(T_1,T_2)$}};
 
 	            \end{tikzpicture}
 	            }
 	            \quad
 	            \subfloat[path B]{\label{path B}
        	    \begin{tikzpicture}[scale = 0.4]
            	    \path[dashed, draw=black]
 	                (0,0) to (10,0)
 	                (10,0) to (10,10)
 	                (10,10) to (0,10)
 	                (0,10) to (0,0);
 	                
 	                \path[line width = 1 pt, draw=black, postaction ={on each segment = {mid arrow}}]
 	                (0,0) to (0,10)
 	                (0,10) to (10,10);
 	                
 	                \node[circle, fill, inner sep=1.25 pt] at (0,0) {};
 	                \node[circle, fill, inner sep=1.25 pt] at (10,10) {};
 	                \node[circle, fill, inner sep=1.25 pt] at (0,10) {};
 	                
 	                \node at (0,-0.5) {\footnotesize{$(0,0)$}};
 			        \node at (0,10.5) {\footnotesize{$(0,T_2)$}};
 			        \node at (10,10.5) {\footnotesize{$(T_1,T_2)$}};
 
 	            \end{tikzpicture}
 	            }
 	            \quad
 	            \subfloat[path C]{\label{path C}
        	    \begin{tikzpicture}[scale = 0.4]
            	    \path[dashed, draw=black]
 	                (0,0) to (10,0)
 	                (10,0) to (10,10)
 	                (10,10) to (0,10)
 	                (0,10) to (0,0);
 	                
 	                \path[line width = 1 pt, draw=black, postaction ={on each segment = {mid arrow}}]
 	                (0,0) to (5,0)
 	                (5,0) to (5,10)
 	                (5,10) to (10,10);
 	                
 	                \node[circle, fill, inner sep=1.25 pt] at (0,0) {};
 	                \node[circle, fill, inner sep=1.25 pt] at (10,10) {};
 	                \node[circle, fill, inner sep=1.25 pt] at (5,0) {};
 	                \node[circle, fill, inner sep=1.25 pt] at (5,10) {};
 	                
 	                \node at (0,-0.5) {\footnotesize{$(0,0)$}};
 			        \node at (5,-0.5) {\footnotesize{$(\tau,0)$}};
 			        \node at (5,10.5) {\footnotesize{$(\tau,T_2)$}};
 			        \node at (10,10.5) {\footnotesize{$(T_1,T_2)$}};
 
 	            \end{tikzpicture}
 	            }
 	            \caption{The first three simple paths}\label{example path}
 	\end{figure}
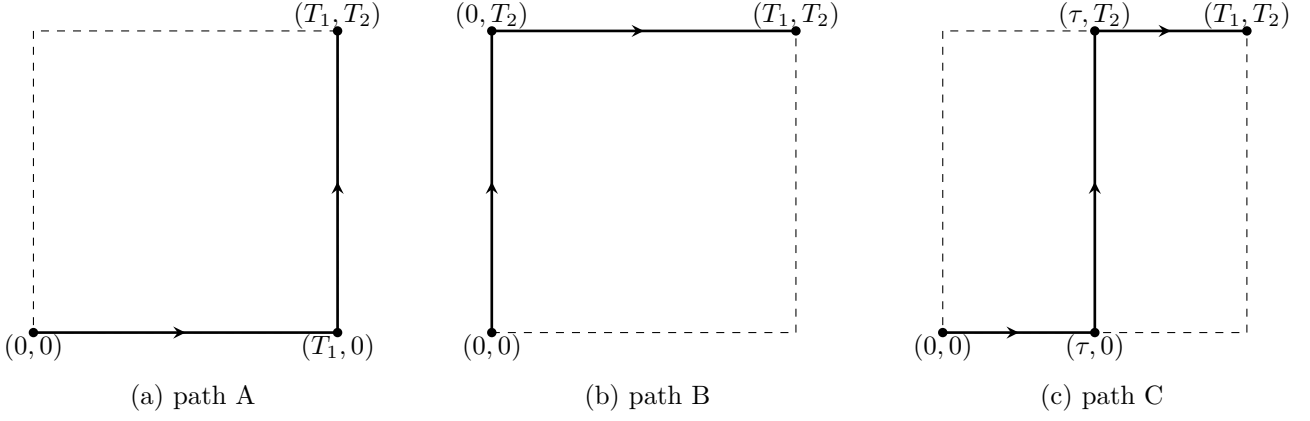
 	\\
 	In figure \ref{path A}, the multi-time propagator reads
 	\begin{equation}
 	    K_A\left( \mathbf{q}(T_1,T_2),(T_1,T_2); \mathbf{q}(0,0),(0,0)\right) = \mathcal{Q}_Ae^{\frac{i}{\hbar}S_A[\mathbf{q}_c]}\;,
 	\end{equation}
 	where
 	\begin{equation}
 	    \mathcal{Q}_A = \int_{\mathbf{y}(0,0)=0}^{\mathbf{y}(T_1,T_2)=0}\mathscr{D}_A[\mathbf{y}(t_1,t_2)]e^{\frac{i}{2\hbar}\left(\int_{(0,0)}^{(T_1,0)} dt_1\left[\left(\frac{\partial \mathbf{y}}{\partial t_1}\right)^2-\omega_1^2\mathbf{y}^2\right]+
        \int_{(T_1,0)}^{(T_1,T_2)}dt_2\left[\left(\frac{\partial \mathbf{y}}{\partial t_2}\right)^2-\omega_2^2\mathbf{y}^2\right]\right)}\;.\label{Q_A}
 	\end{equation}
 	The whole exponent term of the equation \eqref{Q_A} can be written as
 	\begin{align}
 	    &\int_{(0,0)}^{(T_1,0)} dt_1\left[\left(\frac{\partial \mathbf{y}}{\partial t_1}\right)^2-\omega_1^2\mathbf{y}^2\right]+
        \int_{(T_1,0)}^{(T_1,T_2)}dt_2\left[\left(\frac{\partial \mathbf{y}}{\partial t_2}\right)^2-\omega_2^2\mathbf{y}^2\right] \nonumber
        \\&= \int_{(0,0)}^{(T_1,0)} dt_1\mathbf{y}\left[-\left(\frac{\partial}{\partial t_1}\right)^2-\omega_1^2\right]\mathbf{y}+\mathbf{y}(T_1,0)\frac{\partial \mathbf{y}(T_1,0)}{\partial t_1}
        \nonumber\\
        &\;\;\;\;\;\;\;\;\;\;\;\;+\int_{(T_1,0)}^{(T_1,T_2)}dt_2\mathbf{y}\left[-\left(\frac{\partial}{\partial t_2}\right)^2-\omega_2^2\right]\mathbf{y}-\mathbf{y}(T_1,0)\frac{\partial \mathbf{y}(T_1,0)}{\partial t_2}\;.
 	\end{align}
 	For this particular path, the fluctuation $\mathbf{y}$ can be expressed in the form
 	\begin{equation}
 	    \mathbf{y}(t_1,t_2) = \sum_na_n\mathbf{y}_{n,A}(t_1,t_2) = \sum_na_n\left(\sqrt{\frac{2}{T_1}}\sin\left(\frac{n\pi}{T_1}t_1\right)\cos\left(\frac{n\pi}{T_2}t_2\right)+\sqrt{\frac{2}{T_2}}\sin\left(\frac{n\pi}{T_2}t_2\right)\cos\left(\frac{n\pi}{T_1}t_1\right)\right)\;,\label{basisA}
 	\end{equation}
 	where $0\leq t_{1,2}\leq T_{1,2}$ and it is not difficult to show that the orthonormality condition holds
 	\begin{equation}
 	    \int_{(0,0)}^{(T_1,0)}dt_1\mathbf{y}_{n,A}\mathbf{y}_{m,A} = \int_{(T_1,0)}^{(T_1,T_2)}dt_2\mathbf{y}_{n,A}\mathbf{y}_{m,A} = \delta_{nm}\;.
 	\end{equation}
 	Therefore, the equation \eqref{Q_A} becomes
 	\begin{equation}
 	    \mathcal{Q}_A = \int\mathscr{D}[a_n]e^{\frac{i}{2\hbar}\sum_n\abs{a_n}^2\left(-\omega_1^2-\omega_2^2+\left(\frac{n\pi}{T_1}\right)^2+\left(\frac{n\pi}{T_2}\right)^2\right)}\;.\label{rQ_A}
 	\end{equation}
 	Next, we will repeat the same process with the path in figure \ref{path B} and we find that
 	\begin{equation}
 	    \mathcal{Q}_B = \int_{\mathbf{y}(0,0)=0}^{\mathbf{y}(T_1,T_2)=0}\mathscr{D}_B[\mathbf{y}(t_1,t_2)]e^{\frac{i}{2\hbar}\left(
        \int_{(0,0)}^{(0,T_2)}dt_2\left[\left(\frac{\partial \mathbf{y}}{\partial t_2}\right)^2-\omega_2^2\mathbf{y}^2\right]+\int_{(0,T_2)}^{(T_1,T_2)} dt_1\left[\left(\frac{\partial \mathbf{y}}{\partial t_1}\right)^2-\omega_1^2\mathbf{y}^2\right]\right)}\;.\label{Q_B}
 	\end{equation}
 	Fortunately, the eigenbases in the equation \eqref{basisA} are still applicable. Thus, the equation \eqref{Q_B} can be simply reduced to
 	\begin{equation}
 	    \mathcal{Q}_B = \int\mathscr{D}[a_n]e^{\frac{i}{2\hbar}\sum_n\abs{a_n}^2\left(-\omega_1^2-\omega_2^2+\left(\frac{n\pi}{T_1}\right)^2+\left(\frac{n\pi}{T_2}\right)^2\right)}\;.\label{rQ_B}
 	\end{equation}
    For the path in figure \ref{path C}, the $\mathcal{Q}$-factor for multi-time propagator is
 	\begin{equation}
 	    \mathcal{Q}_C = \int_{\mathbf{y}(0,0)=0}^{\mathbf{y}(T_1,T_2)=0}\mathscr{D}_C[\mathbf{y}(t_1,t_2)]e^{\frac{i}{2\hbar}\left(\int_{(0,0)}^{(\tau,0)} dt_1\left[\left(\frac{\partial \mathbf{y}}{\partial t_1}\right)^2-\omega_1^2\mathbf{y}^2\right]+
        \int_{(\tau,0)}^{(\tau,T_2)}dt_2\left[\left(\frac{\partial \mathbf{y}}{\partial t_2}\right)^2-\omega_2^2\mathbf{y}^2\right]+\int_{(\tau,T_2)}^{(T_1,T_2)} dt_1\left[\left(\frac{\partial \mathbf{y}}{\partial t_1}\right)^2-\omega_1^2\mathbf{y}^2\right]\right)}\;.\label{Q_C}
 	\end{equation}
 	The fluctuation $\mathbf{y}$ must be expressed in a new set of eigenbases as
 	\begin{align}
 	    \mathbf{y}(t_1,t_2) = \sum_{n}a_n\mathbf{y}_{n,C}(t_1,t_2)=
	    \begin{cases} \sum_na_n \sqrt{\frac{2}{T_1}}\sin\left(\frac{n\pi}{T_1}t_1\right)\cos\left(\frac{n\pi}{T_2}t_2\right)\;\;;\;\;(t_1\leq\tau\;\text{at}\;t_2=0)\cup(t_1\geq\tau\;\text{at}\;t_2=T_2)\\ 
        \sum_na_n \sqrt{\frac{2}{T_1}}\cos\left(\frac{n\pi}{\tau}t_1\right)\sin\left(\frac{n\pi}{T_2}t_2\right)\;\;;\;\;(t_2\in[0,T_2]\;\text{at}\;t_1=\tau)
        \end{cases}\;.
 	\end{align}
 	We finally obtain the equation \eqref{Q_C} in the form
 	\begin{equation}
 	    \mathcal{Q}_C = \int\mathscr{D}[a_n]e^{\frac{i}{2\hbar}\sum_n\abs{a_n}^2\left(-\omega_1^2-\omega_2^2+\left(\frac{n\pi}{T_1}\right)^2+\left(\frac{n\pi}{T_2}\right)^2\right)}\;.\label{rQ_C}
 	\end{equation}
 	Here we notice that $\mathcal{Q}_A=\mathcal{Q}_B=\mathcal{Q}_C$ and therefore 
 	\begin{equation}
 	    K\left( \mathbf{q}(T_1,T_2),(T_1,T_2); \mathbf{q}(0,0),(0,0)\right) = \mathcal{Q}_Ae^{\frac{i}{\hbar}S_A[\mathbf{q}_c]} = \mathcal{Q}_Be^{\frac{i}{\hbar}S_B[\mathbf{q}_c]} = \mathcal{Q}_Ce^{\frac{i}{\hbar}S_C[\mathbf{q}_c]}\;,
 	\end{equation}
    which is nothing but the path independent feature of the multi-time propagator in case of quadratic Lagrangian 1-forms.
	\newpage
	\section*{Acknowledgements}
	T. Kongkoom is supported by Development and Promotion of Science and Technology Talents Project (DPST).
	\\
	\\
	Conflict of Interest: The authors declare that they have no
	conflicts of interest.
	
	\bibliographystyle{unsrt}
	\bibliography{bibliography.bib}

\end{document}